\documentclass[11pt,a4paper,reqno]{amsart}
\usepackage{amsthm,amsmath,amsfonts,amssymb,amsxtra,bookmark,dsfont,bm}

\theoremstyle{plain}
\newtheorem{theorem}{Theorem}
\newtheorem{lemma}[theorem]{Lemma}

\def\Tr{{\rm Tr}}

\def\bq{\begin{eqnarray}}
\def\eq{\end{eqnarray}}
\def\bqq{\begin{align*}}
\def\eqq{\end{align*}}

\def\nn{\nonumber}

\def\eps{\varepsilon}

\def\Re{\operatorname{Re}}

\def\1{{\ensuremath {\mathds 1} }}

\def\cF {\mathcal{F}}

\def\R {\mathbb{R}}

\def\cN {\mathcal{N}}

\def\cG {\mathcal{G}}

\def\gH{\mathcal{H}}

\def\R {\mathbb{R}}

\def\N {\mathbb{N}}

\def\cC{\mathcal{C}}
\def\cD{\mathcal{D}}

\def\d{{\rm d}}
\def\bT{{\mathbb{T}}}
\def\bZ{{\mathbb{Z}}}

\def\wt{\widetilde}
\def\bH{\mathbb{H}}
\def\dGamma{{\rm d}\Gamma}

\title[The ground state density of a mean-field Bose gas]{Two-term expansion of the ground state one-body density matrix of a mean-field Bose gas}

 \author[P.T. Nam]{Phan Th\`anh Nam}
\address{Department of Mathematics, LMU Munich, Theresienstrasse 39, 80333 Munich, Germany} 
\email{nam@math.lmu.de}
\author[M. Napi\'orkowski]{Marcin Napi\'orkowski}
\address{Department of Mathematical Methods in Physics, Faculty of Physics, University of Warsaw,  Pasteura 5, 02-093 Warszawa, Poland}
\email{marcin.napiorkowski@fuw.edu.pl}

\begin{document}

\begin{abstract}  We consider the homogeneous Bose gas  on a unit torus in the mean-field regime when the interaction strength is proportional to the inverse of the  particle number. In the limit when the number of particles becomes large, we derive a two-term expansion of the one-body density matrix of the ground state. The proof is based on a  cubic correction to  Bogoliubov's approximation of the ground state energy and the ground state.   
\end{abstract}

\date{\today}

\maketitle

\setcounter{tocdepth}{1}

\section{Introduction}

We consider a homogeneous system of $N$ bosons on the unit torus $\bT^d$, for any dimension $d\ge 1$. The system is governed by the mean-field Hamiltonian
\bq \label{eq:HN-intro}
H_{N} = \sum_{j=1}^N -\Delta_{x_j} +\frac{1}{N-1}\sum_{1\le j<k\le N} w(x_j-x_k)
\eq
which acts on the bosonic Hilbert space
$$\gH^N=L^2_{\rm sym}((\bT^d)^N).$$
Here the kinetic operator $-\Delta$ is the usual Laplacian (with periodic boundary conditions). The interaction potential $w$ is a real-valued, even function. We assume that its Fourier transform is non--negative and integrable, namely  
$$
w(x)= \sum_{p\in 2\pi \mathbb{Z}^d} \widehat{w}(p) e^{ip\cdot x} \quad \text{with}\quad  0\le \widehat w \in \ell^1( 2\pi \mathbb{Z}^d ). 
$$
In particular, $w$ is bounded. Since $w$ is even, $\widehat{w}$ is also even.

Under the above conditions, $H_N$ is well defined on the core domain of smooth functions. Moreover, it is well-known that $H_N$ is bounded from below and can be extended to be a self-adjoint operator by Friedrichs' method. The self-adjoint extension, still denoted by $H_N$, has a unique ground state $\Psi_N$ (up to a complex phase)  which solves the variational problem
$$
E_N = \inf_{\| \Psi\|_{\gH^N}=1} \langle \Psi, H_N \Psi\rangle. 
$$
Here $\langle \cdot, \cdot\rangle$ is the inner product in $\gH^N$. \footnote{We use the convention that  $\langle \cdot, \cdot\rangle$ is linear in the second argument and anti-linear in the first.}  

In the present paper, we are interested in the asymptotic behavior of the ground state $\Psi_N \in \gH^N$ of $H_N$ in the limit when $N\to \infty$. More precisely, we will focus on the one-body density matrix $\gamma_{\Psi_N}^{(1)}$ which is a trace class operator on $L^2(\bT^d)$ with kernel
$$
\gamma_{\Psi_N}^{(1)}(x,y) = N  \int_{\bT^{d(N-1)}} \Psi_N (x, x_1, \ldots, x_N) \overline{\Psi_N (y, x_2, \ldots x_N)}dx_2\ldots dx_N.
$$
Note that $\gamma_{\Psi_N}^{(1)}\ge 0$ and $\Tr \gamma_{\Psi_N}^{(1)}=N$. 

\subsection{Main result} 
Our main theorem is  

\begin{theorem}[Ground state density matrix]\label{thm:main-1} Assume that $0\le \widehat w\in \ell^1( (2\pi \mathbb{Z})^d )$. Then the ground state $\Psi_N$ of the Hamiltonian $H_N$ in \eqref{eq:HN-intro} satisfies 
\begin{align*}
\lim_{N\to \infty} \Tr \Big| \gamma^{(1)}_{\Psi_N} - \Big(N-\sum_{p\ne 0} \gamma_p^2\Big) |u_0\rangle \langle u_0| - \sum_{p\ne 0} \gamma_p^2 |u_p\rangle\langle u_p| \Big| =0
\end{align*}
where 
$$
u_p(x)=e^{ip\cdot x},\quad \gamma_{p} = \frac{\alpha_p}{\sqrt{1-\alpha_p^2}}, \quad \alpha_p= \frac{\widehat w(p)}{p^2+ \widehat w(p) + \sqrt{p^4+ 2p^2 \widehat w(p)}} \cdot
$$
\end{theorem}

Here $|u\rangle \langle u|$ is the orthogonal projection on $u$. We use  the bra-ket notation, where $|u\rangle=u$ is a vector in the Hilbert space $\gH$ and $\langle u|$ is an element in the dual space of $\gH$ which maps any vector $v\in \gH$ to the inner product $\langle u,v\rangle_{\gH}$. 

%
%

To the leading order, our result implies Bose-Einstein condensation, namely 
$$
\lim_{N\to \infty}\frac{1}{N}\gamma_{\Psi_N}^{(1)} =  |u_0\rangle \langle u_0|
$$
in the trace norm. This result is  well-known and it follows easily from Onsager's inequality
\bq \label{eq:Onsager}
\frac{1}{N-1} \sum_{1\le i<j\le N} w(x_i-x_j) \ge \frac{N}{2}\widehat w(0) - \frac{N}{N-1} w(0)
\eq
(see \cite{Seiringer-11}). The significance of Theorem \ref{thm:main-1} is that it gives the next order correction to $\gamma^{(1)}_{\Psi_N}$, thus justifying Bogoliubov's approximation in a rather strong sense as we will explain.

\subsection{Bogoliubov's approximation}
It is convenient to turn to the grand canonical setting. Let us introduce the Fock space 
$$
\cF= \bigoplus_{n=0}^\infty \gH^n = \mathbb{C} \oplus \gH \oplus \gH^2 \oplus \cdots  
$$
For any Fock space vector $\Psi=(\Psi_n)_{n=0}^\infty \in \cF$ with $\Psi_n\in \gH^n$, we define its norm  by 
$$
\|\Psi\|_{\cF}^2= \sum_{n=0}^\infty \|\Psi_n \|_{\gH^n}^2. 
$$
and define the particle number expectation by 
$$
\langle \Psi, \cN \Psi\rangle = \sum_{n=0}^\infty n \|\Psi_n\|^2_{\gH^n}. 
$$
In particular, the vacuum state $|0\rangle = (1,0,0,...)$ is a normalized vector on Fock space which has the particle number expectation $\langle 0| \cN |0\rangle =0$. 

For any $f\in \gH$, the creation operator $a^*(f)$ on Fock space maps from $\gH^n$ to $\gH^{n+1}$ for every $n\ge 0$ and satisfies  
\begin{align*}
	(a^* (f) \Psi_n )(x_1,\dots,x_{n+1})&= \frac{1}{\sqrt{n+1}} \sum_{j=1}^{n+1} f(x_j)\Psi(x_1,\dots,x_{j-1},x_{j+1},\dots, x_{n+1}), \quad \forall \Psi_n\in \gH^n. 
\end{align*}
Its adjoint is the annihilation operator $a(f)$, which maps from $\gH^n$ to $\gH^{n-1}$ for every $n\ge 0$ (with convention $\gH^{-1}=\{0\}$) and satisfies  
\begin{align*}
	(a(f) \Psi_n )(x_1,\dots,x_{n-1}) &= \sqrt{n} \int_{\R^d} \overline{f(x_n)}\Psi(x_1,\dots,x_n) d x_n, \quad \forall \Psi_n \in \gH^n. 
\end{align*}
We will denote by $a_p^*$ and $a_p$ the creation and annihilation operators with momentum $p\in 2\pi \mathbb{Z}^d$, namely
$$
a_p^*=a^*(u_p), \quad a_p=a(u_p), \quad u_p(x)=e^{i p \cdot x}.  
$$
They satisfy the canonical commutation relation (CCR)
\bq \label{eq:CCR}
[a_p,a_q]=0=[a_p^*,a_q^*], \quad [a_p,a_q^*]=\delta_{p,q}
\eq
where $[X,Y]=XY-YX.$ 

The creation and annihilation operators can be used to express several operators on Fock space. For example, the number operator can be written as
$$
\cN = \sum_{n=0}^\infty n \1_{\gH^n} = \sum_{p\in 2\pi \bZ^d} a_p^* a_p.
$$
Similarly, the Hamiltonian $H_N$ in \eqref{eq:HN-intro} can be rewritten as 
\bq \label{eq:grand-Hamiltonian}
H_N =\sum_{p \in 2\pi \bZ^d}  p^2 a_p^* a_p +\frac{1}{2(N-1)} \sum_{p,q,k \in \bZ^d} \widehat w(k) a^*_{p-k}a^*_{q+k}a_p a_q.
\eq
The right side of \eqref{eq:grand-Hamiltonian} is an operator on Fock space, which coincides with \eqref{eq:HN-intro} when being restricted to $\gH^N$. In the following we will only use the grand--canonical formula \eqref{eq:grand-Hamiltonian}.

In 1947, Bogoliubov \cite{Bogoliubov-47} suggested a heuristic argument to compute the low-lying spectrum of the operator $H_N$ by using a perturbation around the condensation.  Roughly speaking, he proposed to first substitute all operators $a_0$ and $a_0^{*}$ in \eqref{eq:grand-Hamiltonian} by the scalar number $\sqrt{N}$ (c-number substitution\footnote{Strictly speaking, for $a_0^*a_0^*a_0a_0$ we should rewrite it as $(a_0^*a_0)^2-a_0^*a_0$ before doing the substitution}), and then ignore all interaction terms which are coupled with coefficients of order $o(1)_{N\to \infty}$. All this leads to the formal expression
\bq \label{eq:Bog-formal}
H_N \approx \frac{N}{2}\widehat w(0) + \bH_{\rm Bog}
\eq
where 
\begin{equation}\label{eq:BogHam}
\bH_{\rm Bog}= \sum_{p\ne 0} \left( \big( p^2+\widehat w(p)  \big)  a_p^* a_p + \frac{1}{2}\widehat w(p) \big( a^*_p a^*_{-p}+ a_p a_{-p}\big)\right).
\end{equation}

Note that the expression \eqref{eq:Bog-formal} is formal since $H_N$ acts on the $N$-body Hilbert space $\gH^N$ while the Bogoliubov Hamiltonian $\bH_{\rm Bog}$ acts on the excited Fock space 
$$
\cF_+ = \bigoplus_{n=0}^\infty \gH_+^n = \mathbb{C} \oplus \gH_+ \oplus \gH_+^2 \oplus \cdots, \quad \gH_+ = Q \gH
$$
where we have introduced the projections 
$$ Q  = \sum_{p\ne 0} |u_p\rangle \langle u_p| = 1- P, \quad P= |u_0\rangle \langle u_0|.$$
In particular, unlike $H_N$, the quadratic Hamiltonian $\bH_{\rm Bog}$ does not preserve the number of particles. Nevertheless,  $\bH_{\rm Bog}$ can be {\em diagonalized} by the following unitary transformation on $\cF_+$
\bq \label{eq:UB-intro}
U_{\rm B}= \exp\Big( \sum_{p\neq 0} \beta_p   (a^*_p a^*_{-p}-a_p a_{-p})\Big)
\eq
where the coefficients $\beta_p>0$ are determined by 
$$
 \tanh(2\beta_p)=\alpha_p = \frac{\widehat w(p)}{p^2+ \widehat w(p) + \sqrt{p^4+ 2p^2 \widehat w(p)}} \cdot
$$
In fact, by using the CCR \eqref{eq:CCR} it is straightforward to check that 
\bq \label{eq:Bog-trans}
U_{\rm B} a_p U_{\rm B}^* = \frac{a_p+\alpha_p a_{-p}^*}{\sqrt{1-\alpha_p^2}}=: \sigma_p a_p +\gamma_{p}a_{-p}^*,  \qquad \forall p\ne 0
\eq
where
$$
\sigma_p := \frac{1}{\sqrt{1-\alpha_p^2}} = \cosh(\beta_p), \quad \gamma_p:= \frac{\alpha_p}{\sqrt{1-\alpha_p^2}} = \sinh(\beta_p). 
$$
Consequently, 
\bq \label{eq:Bog-UBHUB}
U_{\rm B} \bH_{\rm Bog} U_{\rm B}^* = E_{\rm Bog} + \sum_{p\ne 0} e(p) a_p^* a_p, 
\eq
where
$$
E_{\rm Bog} = -  \frac{1}{2} \sum_{p\ne 0}   \left( |p|^2 + \widehat w(p) - e(p) \right) , \quad e_p=\sqrt{|p|^4+2 |p|^2 \widehat w(p)}. 
$$
Note that the assumption $0\le \widehat{w} \in \ell^1(2\pi \bZ^d)$ ensures that $E_{\rm Bog}$ is finite. Moreover we have the uniform bounds
\begin{align} \label{eq:sigma-gamma-w}
\sum_{p \ne 0} \gamma_p \le C,\quad \sup_{p\ne 0} \sigma_p \le C . 
\end{align}

Thus  Bogoliubov's approximation predicts that the ground state energy of $H_N$ is 
\bq \label{eq:EN-GSE}
E_{ N}= \frac{N}{2} \widehat w(0) + E_{\rm Bog} + o(1)_{N\to \infty}. 
\eq
In 2011, Seiringer \cite{Seiringer-11} gave the first rigorous proof of \eqref{eq:EN-GSE}. He also proved that the low-lying spectrum of $H_N$ is given approximately by the {\em elementary excitation} $e_p$.  These results have been extended to inhomogeneous trapped systems in \cite{GreSei-13}, to more general interaction potentials in \cite{LewNamSerSol-15}, to a large volume limit in \cite{DerNap-14}, and to situations of multiple-condensation in \cite{NamSei-15,RouSpe-18}. 

\medskip
Let us recall the approach in \cite{LewNamSerSol-15} which also provides the convergence of the  ground state of the mean-field Hamiltonian $H_N$ in \eqref{eq:HN-intro}. Mathematically, the formal expression \eqref{eq:Bog-formal} can be made rigorous using the unitary operator introduced in \cite{LewNamSerSol-15}
$$U_N: \gH^N\to \cF_+^{\le N}=\1^{\le N} \cF_+, \quad \1^{\le N} = \1(\cN_+ \le N) $$ 
which is defined by 
\begin{align} \label{eq:UN}
U_N = \sum_{j=0}^N Q^{\otimes j} \left(\frac{a_0^{N-j}}{\sqrt{(N-j)!}} \right), \quad U_N^*= \bigoplus_{j=0}^N \left(\frac{(a_0^*)^{N-j}}{\sqrt{(N-j)!}} \right).
\end{align}
Recall from  \cite[Proposition 4.2]{LewNamSerSol-15} that 
\begin{align} \label{eq:action}
U_N a_p^* a_q U_N^*&=  a_p^* a_q,\quad U_N a_p^* a_0  U_N^* =  a_p^* \sqrt{N-\cN_+}, \quad \forall p,q \ne 0
\end{align}
where $\cN_+$ is the number operator on the excited Fock space $\cF_+$,
$$
\cN_+ = \sum_{p\ne 0} a_p^* a_p.
$$
Thus $U_N$ implements the c-number substitution in Bogoliubov's argument because it replaces $a_0$ by $\sqrt{N-\cN_+}\approx \sqrt{N}$ (we have $\cN_+\ll N$ due to the condensation). Then the formal expression \eqref{eq:Bog-formal} can be reformulated as
\bq \label{eq:Bog-UN-op}
U_N H_N U_N^* \approx \frac{N}{2}\widehat w(0) + \bH_{\rm Bog}
\eq
which is rigorous since the operators on both sides act on the same excited Fock space. By justifying \eqref{eq:Bog-UN-op}, the authors of \cite{LewNamSerSol-15} recovered the convergence of eigenvalues of $H_N$ first obtained in \cite{Seiringer-11}, and also  obtained the convergence of eigenfunctions of $H_N$ to those of $\bH_{\rm Bog}$. In particular, for the ground state, we have from  \cite[Theorem 2.2]{LewNamSerSol-15} that 
\bq \label{eq:Bog-norm}
\lim_{N\to\infty}U_N \Psi_N = U_B |0\rangle
\eq
where $|0\rangle$ is the vacuum in Fock space. The convergence \eqref{eq:Bog-norm} holds strongly in norm of $\cF_+$, and also strongly in the norm induced by the  quadratic form of $\bH_{\rm Bog}$ in $\cF_+$. In particular, this  implies the convergence of one-body density matrix
\bq \label{eq:Bog-DM}
\lim_{N\to  \infty} Q \gamma_{\Psi_N}^{(1)} Q =  \sum_{p\ne 0} \gamma_p^2 |u_p\rangle \langle u_p|
\eq
in trace class (see  \eqref{eq:QgQ-proof} for a  detailed explanation). 
Since $\Tr \gamma_{\Psi_N}^{(1)}=N$,  \eqref{eq:Bog-DM} is equivalent to 
\begin{align} \label{eq:DM-weak}
\lim_{N\to \infty} \Tr \Big| P \gamma^{(1)}_{\Psi_N} P + Q \gamma^{(1)}_{\Psi_N} Q - \Big(N-\sum_{p\ne 0} \gamma_p^2\Big) |u_0\rangle \langle u_0| - \sum_{p\ne 0} \gamma_p^2 |u_p\rangle\langle u_p| \Big| =0. 
\end{align}
Recall that $P=|u_0\rangle \langle u_0|=1-Q$. The formula \eqref{eq:DM-weak}  looks similar to the result in Theorem \ref{thm:main-1}, except that the cross term $P \gamma_{\Psi_N}^{(1)}Q + Q \gamma_{\Psi_N}^{(1)}P$ is missing. Putting differently, to get the result in Theorem \ref{thm:main-1} we have to show that  
\bq \label{eq:DM-strong}
\lim_{N\to \infty} \Tr \Big| P \gamma_{\Psi_N}^{(1)}Q + Q \gamma_{\Psi_N}^{(1)}P \Big|=0.  
\eq

As explained in \cite[Eq. (2.19)]{LewNamSerSol-15},  from \eqref{eq:Bog-DM} and the Cauchy-Schwarz inequality one only obtains that the left side of \eqref{eq:DM-strong} is of order $O(\sqrt{N})$. Moreover, \eqref{eq:DM-strong} implies that
\bq\label{eq:DM-strong-2}
 \lim_{N\to \infty} \sqrt{N} \langle U_N \Psi_N, a_p U_N \Psi_N\rangle  =0 , \quad \forall p\ne 0,
 \eq 
 thus answering an open question in \cite{Nam-18}. As explained in \cite[Section 5]{Nam-18}, \eqref{eq:DM-strong-2} would follow if we could replace $U_N \Psi_N$ by $U_B |0\rangle$ (which is a {quasi-free state}, and thus satisfies Wick's Theorem \cite[Chapter 10]{Solovej-14}). However, the norm convergence \eqref{eq:Bog-norm} is not strong enough to justify \eqref{eq:DM-strong-2}. 
 
\subsection{Outline of the proof}  To prove Theorem \ref{thm:main-1} we have to extract some information going beyond Bogoliubov's approximation. Roughly speaking, we will refine \eqref{eq:Bog-UN-op} by computing exactly the term of order $O(N^{-1})$. Our proof consists of three main steps. 

\bigskip
\noindent
{\bf Step 1 (Excitation Hamiltonian).} After implementing the c-number substitution, instead of ignoring all terms with coefficients of order $o(1)_{N\to \infty}$, we will keep all terms of order $O(N^{-1})$. More precisely, in Lemma \ref{lem:Bog-app} below we show that 
\begin{align} \label{eq:lem:Bog-app-intro}
U_N H_{N} U_N^*=  \frac{N}{2} \widehat{w}(0)+ \cG_N + O(N^{-3/2})
\end{align}
in an appropriate sense, where
\begin{align*}
\cG_N &= \bH_{\rm Bog}  + \frac{\cN_+(1-\cN_+)}{2(N-1)} \widehat{w}(0)+ \sum_{p\neq 0}   \frac{1-\cN_+}{N-1} \widehat{w}(p) a^*_p a_p + \left(\frac12 \sum_{p\neq 0} \widehat{w}(p) a^*_p a^*_{-p} \frac{1-2\cN_+}{2N}+\,\, \text{h.c.}\right)\\
&\quad + \left(\frac{1}{\sqrt{N}} \sum_{\ell,p\neq 0, \ell+p\neq 0} \widehat{w}(\ell) a^*_{p+\ell} a^*_{-\ell} a_{p} +\,\, \text{h.c.} \right)   + \frac{1}{2(N-1)}\sum_{\substack{k,p \neq 0\\ \ell\neq -p,k}}\widehat{w}(\ell) a^*_{p+\ell} a^*_{k-\ell} a_p a_k. 
\end{align*}

The formula \eqref{eq:lem:Bog-app-intro} is obtained by a direct computation using the actions of $U_N$ as in \cite{LewNamSerSol-15}, plus an expansion of $\sqrt{N-\cN_+}$ and $\sqrt{(N-\cN_+)(N-\cN_+-1)}$ in the regime $\cN_+ \ll N$. The advantage of using $\cG_N$ is that it is well-defined on the full Fock space $\cF_+$. This idea has been used to study the norm approximation for the many-body quantum dynamics in  \cite{BNNS-19}.

\bigskip
\noindent
{\bf Step 2 (Quadratic transformation).} Then we conjugate the operator on the right side of \eqref{eq:lem:Bog-app-intro} by the Bogoliubov transformation $U_B$ in \eqref{eq:UB-intro}. In Lemma \ref{lem:quadratic-transformation} we prove that 
\bq \label{eq:quadratic-transformation-intro}
U_B \cG_N U_B^* = \langle0| U_B \cG_N U_B^* |0\rangle + \sum_{p\ne 0} e(p) a_p^* a_p  +  \cC_N  + R_2 
\eq
where 
$$
\cC_N= \frac{1}{\sqrt{N}}\sum_{\substack{p,q\neq 0\\ p+q\neq 0}}\widehat{w}(p) \Big[ (\sigma_{p+q}\sigma_{-p}\gamma_q +\gamma_{p+q}\gamma_{p}\sigma_q ) a^*_{p+q} a^*_{-p} a^*_{-q} +\rm{h.c.} \Big]
$$
and $R_2$ is an error term whose expectation against the ground state is of order $O(N^{-3/2})$.

Note that in $\cC_N$ we keep only cubic terms with three creation operators or three annihilation operators. These are the most problematic terms. 
All other cubic terms, as well as all quartic terms, are of lower order and can be estimated by the Cauchy--Schwarz inequality (the quartic terms always come with a factor $N^{-1}$ instead of $N^{-1/2}$ and this helps).

As we will see, the energy contribution of the cubic term $\cC_N$ is of order $O(N^{-1})$. Thus \eqref{eq:quadratic-transformation-intro} implies that 
\bq \label{eq:Phi'-E}
E_N = \frac{N}{2}\widehat w(0) +  \Big\langle 0 \Big|  U_B \cG_N U_B^*   \Big| 0 \Big\rangle + O(N^{-1})
\eq
which improves \eqref{eq:EN-GSE}. Moreover, for the ground state we have
\bq \label{eq:Phi'}
\langle U_B U_N \Psi_N,   \cN_+ U_B U_N \Psi_N \rangle \le CN^{-1}
\eq
which in turn implies the norm approximation (up to an appropriate choice of the phase factor for $\Psi_N$)
\bq \label{eq:Phi'-norm}
\| U_N \Psi_N - U_B^* |0\rangle \|_{\cF_+}^2 \le CN^{-1}. 
\eq 
and the following bound on the one-body density matrix 
$$
\lim_{N\to \infty} \Tr \Big| P \gamma_{\Psi_N}^{(1)}Q + Q \gamma_{\Psi_N}^{(1)}P \Big| \le C.
$$
Unfortunately the latter bound is  still weaker than \eqref{eq:DM-strong}. Thus the desired result \eqref{eq:DM-strong} cannot be obtained within Bogoliubov's theory. 

\bigskip
\noindent
{\bf Step 3 (Cubic transformation).} To factor out the energy contribution of the cubic term $\cC_N$ in \eqref{eq:quadratic-transformation-intro}, we will use a cubic transformation. It is given by 
\begin{equation}\label{eq:US-intro}
 U_S = e^{S}, \quad S=  \frac{1}{\sqrt{N}}\sum_{\substack{p,q\neq 0\\ p+q\neq 0}} \eta_{p,q}\Big(a^*_{p+q}a^*_{-p}a^*_{-q}\1^{\leq N}-\1^{\leq N} a_{p+q}a_{-p}a_{-q}\Big) 
\end{equation}
where  
\begin{equation}\label{eq:etadef-intro}
\eta_{p,q}=\frac{\widehat{w}(p)\big(\sigma_{p+q}\sigma_p \gamma_q +\gamma_{p+q}\gamma_p \sigma_q\big)}{e_{p+q}+e_p+e_q}.
\end{equation}
From the assumption $\widehat w \in \ell^1(2\pi \bZ^d)$ and the bounds \eqref{eq:sigma-gamma-w} we have the summability  
\bq \label{eq:US-summability}
\sum_{p,q \ne 0} |\eta_{p,q}| \le C. 
\eq
Here we insert the cut-off $\1^{\leq N}$ in the definition of $U_S$ to make sure that it does not change the particle number operator $\cN_+$ too much; see Lemma \ref{lem:moment-US} for details. 

The choice of the cubic transformation above can be deduced on an abstract level. Consider an operator of the form 
$$
A= A_0 + X
$$
where $X$ stands for some perturbation. Then, in principle, we can remove $X$ by conjugating $A$ with $e^{S}$ provided that
$$
X + [S,A_0]=0
$$
and that $[S,[S,A_0]]=-[S,X]$ is small in an appropriate sense. This can be seen by the simple expansions
$$
e^{S} X e^{-S} = X +  \int_0^1 e^{sS} [X,S] e^{-sS} ds
$$
and
$$
e^{S} A_0 e^{-S} = A_0 + [S,A_0] + \int_0^1 \int_0^t e^{sS} [S,[S,A_0]] e^{-sY} ds dt. 
$$
In our situation, $A_0= \sum_{p\ne 0} e(p) a_p^* a_p$ and $X=\cC_N$, allowing to find $S$ explicitly in \eqref{eq:US-intro}. 

In Lemma \ref{lem:cubic-transformation} we prove that 
$$
U_S U_B \cG_N U_B^* U_S^*= \Big\langle 0 \Big| U_S U_B \cG_N U_B^*  U_S^* \Big| 0 \Big\rangle + \sum_{p\ne 0} e(p) a_p^*a_p + R_3  
$$
with an error term $R_3$ whose expectation against the ground state is of order $O(N^{-3/2})$.  This allows us to obtain the following improvements of \eqref{eq:Phi'-E}, \eqref{eq:Phi'} and \eqref{eq:Phi'-norm}. 

\begin{theorem}[Refined ground state estimates]\label{thm:main-2} Assume that $0\le \widehat w\in \ell^1( (2\pi \mathbb{Z})^d )$. Then the ground state energy of the Hamiltonian $H_N$ in \eqref{eq:HN-intro} satisfies 
$$
E_N=  \frac{N}{2}\widehat w(0)  + \Big\langle 0 \Big| U_S U_B \cG_N U_B^*  U_S^* \Big| 0 \Big\rangle + O(N^{-3/2}).
$$
Moreover, if $\Psi_N$ is the ground state of $H_N$, then $\Phi= U_S U_B U_N \Psi_N$ satisfies 
$$
\langle \Phi, \cN_+ \Phi\rangle \le CN^{-3/2}. 
$$
Consequently, we have the norm approximation (up to an appropriate choice of the phase factor for $\Psi_N$)
$$
\| U_N \Psi_N - U_B^* U_S^* |0\rangle \|_{\cF_+}^2 \le CN^{-3/2}. 
$$
\end{theorem}

As we will explain, Theorem \ref{thm:main-2} implies  \eqref{eq:DM-strong} and thus justifies Theorem \ref{thm:main-1}.

The idea of using cubic transformations has been developed to handle dilute Bose gases in \cite{YauYin-09,BBCS-19,BBCS-20,RadSch-19,ABS-20}, where the interaction potential has a much shorter range but the interaction strength is much larger in its range. In this case, the contribution of the cubic terms is much bigger, and Bogoliubov's approximation has to be modified appropriately to capture the short-range scattering effect. Results similar to  \eqref{eq:EN-GSE} have been proved recently for the Gross-Pitaevskii limit \cite{BBCS-19} and for the thermodynamic limit \cite{YauYin-09,FS19}. It is unclear to us how to extend Theorem \ref{thm:main-1} to the dilute regime. 

Our work shows that in the mean-field regime, in contrast to the dilute regime, the cubic terms are smaller, and they actually contribute only to  the next order correction to Bogoliubov's approximation (there are also some quadratic and quartic terms which contribute to the same order of the cubic term). On the other hand, it is interesting that the contribution of the cubic terms is not visible in the expansion of the one-body density matrix in Theorem \ref{thm:main-1}; putting differently the approximation in Theorem \ref{thm:main-1} can be guessed using only Bogoliubov's theory (although its proof requires more information). 

\medskip

There have been also remarkable works concerning higher order expansions in powers of $N^{-1}$ in the mean-field regime; see \cite{Pizzo-15} for a study of the ground state, \cite{BPS-20} for the low-energy spectrum, and \cite{BPPS-19,BPPS-20} for the quantum dynamics. These works are based on perturbative approaches which are very different from ours. Note that the method of Bossmann, Petrat and Seiringer in \cite{BPS-20} also gives access to the higher order expansion of the reduced density matrices   (see \cite[Eq. (3.15)]{BPS-20} for a comparison). We hope that our rather explicit strategy complements  the previous analysis in \cite{Pizzo-15,BPPS-19,BPPS-20,BPS-20} concerning the correction to Bogoliubov's theory in the mean-field regime.

\medskip

\subsection*{Organization of the paper} In Section  \ref{sec:N} we will derive some useful estimates for the particle number operator $\cN_+$. Then we analyze the actions of the transformations $U_N$, $U_B$, $U_S$ in Sections \ref{sec:UN}, \ref{sec:UB}, \ref{sec:US}, respectively. Finally, we prove Theorem \ref{thm:main-2} in Section \ref{sec:Thm-1} and conclude Theorem \ref{thm:main-1} in Section \ref{sec:Thm-1}.

\subsection*{Acknowledgments} We thank Robert Seiringer and Nicolas Rougerie for helpful discussions. The research is funded by the Polish-German Beethoven Classic 3 project ``Mathematics of many-body quantum systems". PTN acknowledges the support from the Deutsche Forschungsgemeinschaft (DFG project Nr. 426365943).  MN acknowledges the support from the National Science Centre (NCN project Nr. 2018/31/G/ST1/01166).

\section{Moment estimates for the particle number operator} \label{sec:N}

In this section we justify the Bose-Einstein condensation by showing that the ground state has a bounded number of excited particles. As explained in \cite{Seiringer-11}, the uniform bound on the expectation of $\cN_+$ follows easily from Onsager's inequality \eqref{eq:Onsager}. For our purpose, we will need uniform bounds for higher moments of $\cN_+$. The following lemma is an extension of \cite[Lemma 5]{Nam-18}.

\begin{lemma}[Number of excited particles] \label{lem:cN+}   If $\Psi_N$ is the ground state of $H_{N}$, then 
$$ \langle \Psi_N, \cN_+^s \Psi_N \rangle \le C_s, \qquad \forall  s\in \N.$$ 
\end{lemma}

\begin{proof} As in \cite[Lemma 5]{Nam-18}, from the operator inequality  
\bq \label{eq:leading-lower}
H_{N}  \ge (2\pi)^{2} \cN_+ +\frac{ N^2}{2(N-1)}\hat{w}(0) -  \frac{N}{2(N-1)} w(0) 
\eq
we obtain 
\bq \label{eq:bound-cN+}
|E_{N}| \le \frac{N}{2}\hat{w}(0) \quad \text{and} \quad\langle \Psi_N, \cN_+^{s} \Psi_N \rangle \le C
\eq
for $s=1,2,3$. Let us assume  that  $s\in\N$ is even. We will show that $\langle \Psi_N, \cN_+^{s+1} \Psi_N \rangle \le C$. 

Since $\Psi_N$ is a ground state of $H_N$, it solves the Schr\"odinger equation 
$$ H_{N}\Psi_N= E_{N}\Psi_N.$$
Consequently, we get the identity 
\begin{align} \label{eq:identity-cN+-HN}
\left\langle \Psi_N, \cN_+^{\frac{s}{2}} \Big( H_{N} - E_{N}  \Big) \cN_+^{\frac{s}{2}} \Psi_N \right\rangle =  \Big\langle \Psi_N,   \cN_+^{\frac{s}{2}} [H_{N}, \cN_+^{\frac{s}{2}}]   \Psi_N \Big\rangle.
\end{align}
The left side of \eqref{eq:identity-cN+-HN} can be estimated using \eqref{eq:leading-lower} and \eqref{eq:bound-cN+} as
\bq \label{eq:identity-cN+-HN-left}
 \left\langle \Psi_N, \cN_+^{\frac{s}{2}} \Big( H_{N} - E_{N}  \Big) \cN_+^{\frac{s}{2}} \Psi_N \right\rangle \ge \left\langle \Psi_N, \Big( (2\pi)^{2}   \cN_+^{s+1} -C\cN_+^s \Big) \Psi_N \right\rangle.
\eq
For the right side of \eqref{eq:identity-cN+-HN}, since
$$[A,B^k]=\sum_{j=0}^{k-1}B^j [A,B] B^{k-j-1},$$
 using \eqref{eq:grand-Hamiltonian} and the CCR \eqref{eq:CCR} we write 
\begin{align} \label{eq:identity-cN+-HN-right-0}
 \cN_+^{\frac{s}{2}}[H_{N},& \cN_+^{\frac{s}{2}}] =\frac{1}{2(N-1)} \sum_{j=0}^{\frac{s}{2}-1} \sum_{\ell \ne 0} \sum_{p,q } \widehat w(\ell) \cN_+^{\frac{s}{2}+j}[a^*_{p-\ell}a^*_{q+\ell}a_p a_q, \cN_+]\cN_+^{\frac{s}{2}-j-1} \nn\\
&=\frac{1}{2(N-1)} \sum_{j=0}^{\frac{s}{2}-1}\sum_{\ell \ne 0} \widehat w(\ell) \cN_+^{\frac{s}{2}+j} \Big( 2 a^*_0 a^*_0 a_{\ell} a_{-\ell} -2 a^*_{-\ell} a^*_{\ell} a_0 a_0   \Big)\cN_+^{\frac{s}{2}-j-1} \nn\\
& + \frac{1}{2(N-1)} \sum_{j=0}^{\frac{s}{2}-1}\sum_{\ell \ne 0 \ne p \ne \ell} \widehat w(\ell) \cN_+^{\frac{s}{2}+j} \Big(a^*_{p-\ell}a^*_{0}a_p a_{-\ell}  - a^*_{p-\ell}a^*_{\ell}a_p a_0  \Big)\cN_+^{\frac{s}{2}-j-1} \nn\\
& + \frac{1}{2(N-1)}\sum_{j=0}^{\frac{s}{2}-1} \sum_{\ell \ne 0 \ne q \ne -\ell} \widehat w(\ell) \cN_+^{\frac{s}{2}+j} \Big(a^*_{0}a^*_{q+\ell}a_{\ell} a_{q}  - a^*_{-\ell}a^*_{q+\ell}a_0 a_q  \Big)\cN_+^{\frac{s}{2}-j-1}.
\end{align}
Now we take the expectation against $\Psi_N$ and estimate. For the first term on the right side of \eqref{eq:identity-cN+-HN-right-0}, by the Cauchy--Schwarz inequality, we get for a given $j$
\begin{align*}
& \left| \left\langle \Psi_N, \sum_{\ell \ne 0} \widehat w(\ell) \cN_+^{\frac{s}{2}+j}  a^*_0 a^*_0 a_{\ell} a_{-\ell}\cN_+^{\frac{s}{2}-j-1} \Psi_N \right\rangle \right| \\
&\qquad = \left| \left\langle \Psi_N, \sum_{\ell \ne 0} \widehat w(\ell) \cN_+^{\frac{s}{2}+j}  a^*_0 a^*_0 (\cN_++1)^{-j}(\cN_++1)^{j} a_{\ell} a_{-\ell}\cN_+^{\frac{s}{2}-j-1} \Psi_N \right\rangle \right|\\ 
&\qquad \le \sum_{\ell \ne 0}   \left\| (\cN_++1)^{-j} a_0 a_0 \cN_+^{\frac{s}{2}+j}   \Psi_N \right\| | \widehat w(\ell)| \left\|(\cN_+ +1)^{j} a_{\ell} a_{-\ell} \cN_+^{\frac{s}{2}-j-1}  \Psi_N\right\| \\
&\qquad = \sum_{\ell \ne 0}   \left\|  a_0 a_0 \cN_+^{\frac{s}{2}}   \Psi_N \right\| | \widehat w(\ell)| \left\| a_{\ell} a_{-\ell} (\cN_+ -1)^{j} \cN_+^{\frac{s}{2}-j-1}  \Psi_N\right\|\\
&\qquad \le \left\|  a_0 a_0 \cN_+^{\frac{s}{2}}    \Psi_N \right\|   \left( \sum_{\ell \ne 0} |\widehat w(\ell)|^2 \right)^{1/2} \left( \sum_{\ell \ne 0}  \left\| a_{\ell} a_{-\ell}(\cN_+ -1)^{j} \cN_+^{\frac{s}{2}-j-1}   \Psi_N\right\| ^2 \right)^{1/2}\\
&\qquad \leq CN \langle \Psi_N, \cN_+^{s} \Psi_N\rangle.
\end{align*}
Here we have used that $a_0 a_0$ commutes with $\cN_+$, that  $a_0^*a_0 \le N$ on $\gH^N$ and that $\sum |\widehat w(\ell)|^2=\|w\|_{L^2}^2<\infty$. Similarly, for the second term, we have 
\begin{align*}
&\left| \left\langle \Psi_N, \sum_{\ell \ne 0} \widehat w(\ell) \cN_+^{\frac{s}{2}+j}   a^*_{-\ell} a^*_\ell a_0 a_0 \cN_+^{\frac{s}{2}-j-1}  \Psi_N \right\rangle  \right| \\
&\quad = \left| \left\langle \Psi_N, \sum_{\ell \ne 0} \widehat w(\ell) \cN_+^{\frac{s}{2}+j}   a^*_{-\ell} a^*_\ell(\cN_+ +1)^{-j-1}(\cN_+ +1)^{j+1} a_0 a_0 \cN_+^{\frac{s}{2}-j-1}  \Psi_N \right\rangle  \right| \\
& \leq \sum_{\ell \ne 0}   \left\| (\cN_++1)^{-j-1} a_{-\ell} a_{\ell} \cN_+^{\frac{s}{2}+j}   \Psi_N \right\| | \widehat w(\ell)| \left\|(\cN_+ +1)^{j+1} a_{0} a_{0} \cN_+^{\frac{s}{2}-j-1}  \Psi_N\right\| \\ & \leq CN \langle \Psi_N, (\cN_++1)^{s} \Psi_N\rangle
\end{align*}
as before. For the third term, we can bound
\begin{align*}
&\left| \left\langle \Psi_N, \sum_{\ell \ne 0 \ne p \ne \ell}  \widehat w(\ell) \cN_+^{\frac{s}{2}+j} a^*_{p-\ell}a^*_{0}a_p a_{-\ell} \cN_+^{\frac{s}{2}-j-1} \Psi_N \right\rangle \right| \\ 
&\le \sum_{\ell \ne 0 \ne p \ne \ell} |\widehat w(\ell)| \|(\cN_++1)^{-j} a_{0} a_{p-\ell} \cN_+^{\frac{s}{2}+j} \Psi_N\| \|(\cN_++1)^{j}  a_p a_{-\ell}\cN_+^{\frac{s}{2}-j-1}   \Psi_N \|  \\
&\le \left( \sum_{\ell \ne 0 \ne p \ne \ell}  |\widehat w(\ell)|^2 \left\| a_{0} a_{p-\ell} \cN_+^{\frac{s}{2}} \Psi_N \right\|^2  \right)^{1/2} \left(   \sum_{\ell \ne 0 \ne p \ne \ell} \left\| a_p a_{-\ell}(\cN_+-1)^{-j}    \cN_+^{\frac{s}{2}-j-1}    \Psi_N \right\|^2 \right)^{1/2} \\
& \le C N^{1/2}\langle \Psi_N, \cN_+^s \Psi_N \rangle^{1/2} \langle \Psi_N, \cN_+^{s+1} \Psi_N\rangle^{1/2}
\end{align*}
and proceed similarly for other terms. Thus in summary, from \eqref{eq:identity-cN+-HN-right-0} we get
\begin{align} \label{eq:identity-cN+-HN-right}
\left| \Big\langle \Psi_N, \cN_+^{\frac{s}{2}}[H_{N},\cN_+^{\frac{s}{2}}] \Psi_N \Big\rangle \right| &\le C \langle \Psi_N, (\cN_+ +1)^s \Psi_N\rangle \nn\\
&\quad + C N^{-1/2}\langle \Psi_N, \cN_+^{s+1} \Psi_N \rangle^{1/2} \langle \Psi_N, \cN_+^s \Psi_N\rangle^{1/2}.
\end{align}
 Inserting \eqref{eq:identity-cN+-HN-left} and \eqref{eq:identity-cN+-HN-right} into \eqref{eq:identity-cN+-HN}, we obtain
\begin{align*}
\left\langle \Psi_N, \Big( (2\pi)^{2}   \cN_+^{s+1} - C\cN_+^s \Big) \Psi_N \right\rangle &\le C \langle \Psi_N, (\cN_+ +1)^s \Psi_N\rangle \nn\\
&\quad + C N^{-1/2}\langle \Psi_N, \cN_+^{s+1} \Psi_N \rangle^{1/2} \langle \Psi_N, \cN_+^s \Psi_N\rangle^{1/2}.
\end{align*}
By the Cauchy-Schwarz inequality
\bq \label{eq:cN^sbound}
\left\langle \Psi_N, \cN_+^s \Psi_N \right\rangle \le 
\left\langle \Psi_N, \cN_+^{s-1}  \Psi_N \right\rangle^{1/2} \left\langle \Psi_N, \cN_+^{s+1}  \Psi_N \right\rangle^{1/2}
\eq
 we get 
 \begin{align*}
\left\langle \Psi_N,   \cN_+^{s+1} \Psi_N \right\rangle &\le C \left\langle \Psi_N, \cN_+^{s-1}  \Psi_N \right\rangle^{1/2} \left\langle \Psi_N, \cN_+^{s+1}  \Psi_N \right\rangle^{1/2} \nn\\
&\quad + C N^{-1/2}\langle \Psi_N, \cN_+^{s+1} \Psi_N \rangle^{3/4} \langle \Psi_N, \cN_+^{s-1} \Psi_N\rangle^{1/4}
\end{align*}
which implies
\begin{align*}
\left\langle \Psi_N,   \cN_+^{s+1} \Psi_N \right\rangle^{1/2} &\le C \left\langle \Psi_N, \cN_+^{s-1}  \Psi_N \right\rangle^{1/2} \nn\\
&\quad + C N^{-1/2}\langle \Psi_N, \cN_+^{s+1} \Psi_N \rangle^{1/4} \langle \Psi_N, \cN_+^{s-1} \Psi_N\rangle^{1/4}. 
\end{align*}
We can now use
$$\left\langle \Psi_N, \cN_+^{s-1} \Psi_N \right\rangle \le 
\left\langle \Psi_N, \cN_+^{s-3}  \Psi_N \right\rangle^{1/2} \left\langle \Psi_N, \cN_+^{s+1}  \Psi_N \right\rangle^{1/2}$$
and obtain
\begin{align*}
\left\langle \Psi_N,   \cN_+^{s+1} \Psi_N \right\rangle^{1/4} &\le C \left\langle \Psi_N, \cN_+^{s-3}  \Psi_N \right\rangle^{1/4} \nn\\
&\quad + C N^{-1/2}\langle \Psi_N, \cN_+^{s+1} \Psi_N \rangle^{1/8}\left\langle \Psi_N, \cN_+^{s-3}  \Psi_N \right\rangle^{1/8} .
\end{align*}
Telescoping this inequality and using \cite[Lemma 5]{Nam-18} we arrive at a bound on $\left\langle \Psi_N,   \cN_+^{s+1} \Psi_N \right\rangle$ that is uniform in $N$. This gives the desired result for odd powers of $\cN_+$. Finally, using \eqref{eq:cN^sbound}, we obtain the bound for any $s\in \N$ and this ends the proof.
\end{proof}

In order to put Lemma \ref{lem:cN+} in a good use, we will also need the fact that the moments of $\cN_+$ are essentially stable under the actions of the Bogoliubov transformation and the cubic transformation.

\begin{lemma}\label{lem:moment-UB} Let $U_B$ be given in  \eqref{eq:UB-intro}. Then  
\bq \label{eq:moment-UB}
U_{\rm B}\cN_+^{k}U_{\rm B}^* \le C_k (\cN_++1)^{k}, \quad \forall k\in \mathbb{N}.  
\eq
\end{lemma}

\begin{lemma} \label{lem:moment-US}
Let $U_S=e^S$ be given in \eqref{eq:US-intro}. Then  for all $t\in [-1,1]$ and $k\in \mathbb{N}$, 
\begin{equation}\label{eq:moment-US}
e^{tS} (\cN_+ +1)^k e^{-tS} \leq C_k (\cN_+ +1)^k.
\end{equation}
\end{lemma}

The results in Lemma \ref{lem:moment-UB} and Lemma \ref{lem:moment-US} are well-known. For the completeness, let us quickly explain the proof of Lemma \ref{lem:moment-US}, following the strategy in \cite[Proposition 4.2]{BBCS-19} (the proof of Lemma \ref{lem:moment-UB} is similar and simpler).

\begin{proof}[Proof of Lemma \ref{lem:moment-US}] Take a normalized vector  $\Phi \in \cF_+$ and define 
$$f(t)=\langle \Phi, e^{tS}(\cN_+ +1)^k e^{-tS}\Phi \rangle, \quad \forall t\in [-1,1].$$
Then 
\begin{align*}
\partial_t f(t)&=\langle \Phi, e^{tS}[S, (\cN_+ +1)^k] e^{-tS}\Phi \rangle\\
&= \frac{2}{\sqrt{N}}\Re \Big\langle \Phi, e^{tS} \sum_{\substack{p,q\neq 0\\ p+q\neq 0}} \eta_{p,q} a^*_{p+q}a^*_{-p}a^*_{-q} \1^{\leq N}\Theta_k(\cN_+)e^{-tS}\Phi \Big\rangle
\end{align*}
with
$$
\Theta_k(\cN_+)= (\cN_+ +1)^k-(\cN_+ +4)^k.
$$
Here we have used 
$$
[a^*_{p+q}a^*_{-p}a^*_{-q}\1^{\leq N}, (\cN_+ +1)^k]=a^*_{p+q}a^*_{-p}a^*_{-q}\1^{\leq N} \Theta_k(\cN_+). 
$$
It is obvious that $|\Theta_k(\cN_+)| \le C_k (\cN_+ +1)^{k-1}$. Combining with the summability   \eqref{eq:US-summability} and the Cauchy-Schwarz inequality we obtain
\begin{align}
\Big| \partial_t f(t) \Big| &\le  \frac{2}{\sqrt{N}} \Big( \sum_{\substack{p,q\neq 0\\ p+q\neq 0}} \Big\|  (\cN_+ +1)^{(k-3)/2} a_{p+q}a_{-p}a_{-q} e^{-tS}\Phi \|^2 \Big)^{1/2} \times \nn\\
&\quad \times \Big( \sum_{\substack{p,q\neq 0\\ p+q\neq 0}} |\eta_{p,q}|^2  \Big\|  \1^{\le N} (\cN_+ +1)^{(3-k)/2} \Theta_k(\cN_+)e^{-tS}\Phi \Big\|^2 \Big)^{1/2}\nn\\
&\le \frac{C_k}{\sqrt{N}} \Big\| \cN_+^{k/2}  e^{-tS} \Phi \Big\| \Big\|  \1^{\le N}(\cN_+  +1) ^{(k+1)/2}  e^{-tS} \Phi \Big\|.  \label{eq:US-cN-pre}
\end{align}
Thanks to the cut-off, we can bound
$$
 \1^{\le N}(\cN_+ +1) ^{(k+1)/2}  \le \sqrt{N+1} \cN_+^{k/2}. 
$$
Thus \eqref{eq:US-cN-pre} implies that  
\begin{align*}
\Big| \partial_t f(t) \Big| \le C_k \Big\|  (\cN_+ +1 )^{(k+1)/2}  e^{-tS} \Phi \Big\|^2 = C_k f(t)
\end{align*}
From Gr\"onwall's lemma, it follows that
\bq \label{eq:ft-bound}
f(t) \le C_k f(0), \quad \forall t\in [-1,1].
\eq
Since the latter bound is uniform in $\Phi$, we get the desired operator inequality.   
\end{proof}

We will also need the following refinement of Lemma \ref{lem:moment-US}.

\begin{lemma} \label{lem:moment-US-2}
Let $U_S=e^S$ be given in \eqref{eq:US-intro}. Then  for all $t\in [-1,1]$ and $k\in \mathbb{N}$, 
\begin{equation}\label{eq:moment-US-2}
e^{tS} \cN_+^k  e^{-tS} \leq C_k \Big(  \cN_+^k + \frac{(\cN_+ +1)^{k+1}}{N} \Big) . 
\end{equation}
\end{lemma}

\begin{proof} Take a normalized vector  $\Phi \in \cF_+$ and define 
$$g(t)=\langle \Phi, e^{tS}\cN_+^k e^{-tS}\Phi \rangle, \quad \forall t\in [-1,1].$$
Then proceeding similarly to \eqref{eq:US-cN-pre}, we have 
$$
\Big| \partial_t g(t) \Big| \le \frac{C_k}{\sqrt{N}} \Big\| \cN_+^{k/2}  e^{-tS} \Phi \Big\| \Big\|  \1^{\le N}\cN_+^{(k+1)/2}  e^{-tS} \Phi \Big\| \le   \frac{C_k}{\sqrt{N}} \sqrt{g(t) f(t)} 
$$
with $f(t)$ being defined in the proof of Lemma \ref{lem:moment-US}.  Using \eqref{eq:ft-bound} and the Cauchy-Schwarz inequality we obtain 
$$
\Big| \partial_t g(t) \Big| \le C_k \Big( g(t) + \frac{\langle \Phi, ( \cN_+ +1 )^{k+1} \Phi \rangle}{N} \Big),  \quad \forall t\in [-1,1].
$$
From Gr\"onwall's lemma, it follows that
$$
g(t) \le  C_k \Big( g(0) +  \frac{1}{N}\langle \Phi, ( \cN_+ +1 )^{k+1} \Phi \rangle  \Big), \quad \forall t\in [-1,1].
$$
The latter bound is uniform in $\Phi$ and it implies the desired conclusion.
\end{proof}

\section{Excitation Hamiltonian}\label{sec:UN}

In this section, we study the action of the transformation $U_N$ in \eqref{eq:UN}. By conjugating $H_N$ with $U_N$, we can factor out the contribution of the condensation. More precisely, we have 
 
\begin{lemma}\label{lem:Bog-app} We have the operator identity on $\cF_{+}^{\le N}$
\begin{align*}
U_N H_{N} U_N^*&=  \frac{N}{2} \widehat{w}(0)+ \1^{\le N} (\cG_N + R_1) \1^{\le N}
\end{align*}
where
\begin{align*}
\cG_N &= \bH_{\rm Bog}  + \frac{\cN_+(1-\cN_+)}{2(N-1)} \widehat{w}(0)+ \sum_{p\neq 0}   \frac{1-\cN_+}{N-1} \widehat{w}(p) a^*_p a_p + \left(\frac12 \sum_{p\neq 0} \widehat{w}(p) a^*_p a^*_{-p} \frac{1-2\cN_+}{2N}+\,\, \text{h.c.}\right)\\
&\quad + \left(\frac{1}{\sqrt{N}} \sum_{\ell,p\neq 0, \ell+p\neq 0} \widehat{w}(\ell) a^*_{p+\ell} a^*_{-\ell} a_{p} +\,\, \text{h.c.} \right)   + \frac{1}{2(N-1)}\sum_{\substack{k,p \neq 0\\ \ell\neq -p,k}}\widehat{w}(\ell) a^*_{p+\ell} a^*_{k-\ell} a_p a_k 
\end{align*}
and the error term $R_1$ satisfies the quadratic form estimate  
\begin{align*}
&\pm R_1 \le \frac{C(\cN_++1)^{3}}{N^{3/2}}.
\end{align*}
Moreover, we have the operator inequality on $\cF_{+}$
\begin{align*}
\1^{\le N} U_N H_{N} U_N^* \1^{\le N} \le  \frac{N}{2} \widehat{w}(0)+ \cG_N + \frac{C(\cN_++1)^{3}}{N^{3/2}}.
\end{align*}
\end{lemma}

\begin{proof}
A straightforward computation using the relations \eqref{eq:UN} shows that
\begin{align*}
U_N H_{N} U_N^*&= \frac{N}{2} \widehat{w}(0)+\frac{\cN_+(1-\cN_+)}{2(N-1)} \widehat{w}(0)+ \sum_{p\neq 0} \Big( p^2 +  \frac{N-\cN_+}{N-1} \widehat{w}(p) \Big) a^*_p a_p \\
&\quad+\frac12 \left(\sum_{p\neq 0} \widehat{w}(p) a^*_p a^*_{-p} \frac{\sqrt{(N-\cN_+)(N-\cN_+ -1)}}{N-1}+\,\, \text{h.c.}\right)\\
&\quad + \left( \sum_{\ell,p\neq 0, \ell+p\neq 0} \widehat{w}(\ell) a^*_{p+\ell} a^*_{-\ell} a_{p} \frac{\sqrt{N-\cN_+}}{N-1}+\,\, \text{h.c.} \right) \\
&\quad + \frac{1}{2(N-1)}\sum_{\substack{k,p \neq 0\\ \ell\neq -p,k}}\widehat{w}(\ell) a^*_{p+\ell} a^*_{k-\ell} a_p a_k.
\end{align*}
This operator identity holds on $\cF_+^{\le N}$. For further analysis, we will expand $\sqrt{N-\cN_+}$ and $\sqrt{(N-\cN_+)(N-\cN_+ -1)}$, making the effective expressions well-defined on the whole Fock space $\cF_+$. This idea has been used before in \cite{BNNS-19}. Here it suffices to use 
\bq \label{eq:Expand-1}
\Big| \frac{\sqrt{N-\cN_+}}{N-1} - \frac{1}{N^{1/2}}\Big| \le \frac{C(\cN_++1)}{N^{3/2}}
\eq
and 
\bq \label{eq:Expand-2}
\Big| \frac{\sqrt{(N-\cN_+)(N-\cN_+-1)}}{N-1} - 1 - \frac{1-2\cN_+}{2N} \Big| \le \frac{C(\cN_++1)^2}{N^2}.
\eq
The operator inequalities \eqref{eq:Expand-1} and \eqref{eq:Expand-2} hold on $\cF_+^{\le N}$. Thus we can write 
$$U_N H_{N} U_N^*=  \frac{N}{2} \widehat{w}(0)+ \bH_{\rm Bog} + \cG_N + R_1$$
with $ \cG_N$ given in the statement of Lemma \ref{lem:Bog-app} and with the error term $R_1=R_{1a}+R_{1b}$ where
\begin{align*}
R_{1a} &= \frac12 \sum_{p\neq 0} \widehat{w}(p) a^*_p a^*_{-p} \Big( \frac{\sqrt{(N-\cN_+)(N-\cN_+-1)}}{N-1} - 1 - \frac{1-2\cN_+}{2N}  \Big)\,\,+  \text{h.c.} ,\\
R_{1b} &= \sum_{\ell,p\neq 0, \ell+p\neq 0} \widehat{w}(\ell) a^*_{p+\ell} a^*_{-\ell} a_{p} \Big( \frac{\sqrt{N-\cN_+}}{N-1} - \frac{1}{\sqrt{N}} \Big) + \,\, \text{h.c.} .
\end{align*}
By the Cauchy-Schwarz inequality, we have the quadratic form estimates
\begin{align*}
\pm R_{1a}&\le N^{-2}  \sum_{p\ne 0}  a^*_p a^*_{-p} (\cN_+ +1) a_{-p}a_p + N^{2} \sum_{p\ne 0} |\widehat{w}(p)|^2 (\cN_+ +1)^{-1/2}\times \\
&\quad \times \Big( \frac{\sqrt{(N-\cN_+)(N-\cN_+-1)}}{N-1} - 1 - \frac{1-2\cN_+}{2N}  \Big)^2 (\cN_+ +1)^{-1/2}\\
&\le \frac{C(\cN_++1)^3}{N^2}
\end{align*}
and 
\begin{align*}
\pm R_{1b}&\le N^{-3/2}  \sum_{\ell,p\neq 0, \ell+p\neq 0} a^*_{p+\ell} a^*_{-\ell} a_{-\ell} a_{p+\ell} \\
&\quad + N^{3/2} \sum_{\ell,p\neq 0, \ell+p\neq 0}   |\widehat{w}(\ell)|^2 \Big( \frac{\sqrt{N-\cN_+}}{N-1} - \frac{1}{\sqrt{N}} \Big) a_p^* a_p  \Big( \frac{\sqrt{N-\cN_+}}{N-1} - \frac{1}{\sqrt{N}} \Big) \\
&\le \frac{C(\cN_++1)^3}{N^{3/2}}.
\end{align*}
This completes the first part of Lemma \ref{lem:Bog-app}.  

Now let us turn to the operator inequality on the Fock space $\cF_+$. We have proved that 
\bq \label{eq:UN-00-00}
\1^{\le N} U_N H_N U_N^* \1^{\le N} \le \1^{\le N} \Big( \frac{N}{2}\widehat w(0) + \cG_N + \frac{C(\cN_++1)^3}{N^{3/2}} \Big) \1^{\le N}. 
\eq
Let us compare the right side of \eqref{eq:UN-00-00} with the corresponding version without the cut-off $\1^{\le N}$. First, consider  the terms commuting with $\cN_+$. Since 
$$
\frac{N}{2}\widehat w(0) + \sum_{p\ne 0} |p|^2 a_p^* a_p + \frac{1}{2(N-1)}\sum_{\substack{k,p \neq 0\\ \ell\neq -p,k}}\widehat{w}(\ell) a^*_{p+\ell} a^*_{k-\ell} a_p a_k  +  \frac{C(\cN_++1)^3}{N^{3/2}} \ge 0,
$$
this operator is not smaller than its product with the cut-off $\1^{\le N}$. Moreover, using 
$$\1^{>N}= \1 - \1^{\le N} =\1(\cN_+ >N) \le \frac{\cN_+}{N}$$
we have
\begin{align*}
\pm \1^{>N}  \Big( \frac{\cN_+(1-\cN_+)}{2(N-1)} \widehat{w}(0) + \sum_{p\neq 0}   \frac{1-\cN_+}{N-1} \widehat{w}(p) a^*_p a_p \Big) \le  \1^{>N}   \frac{C(\cN_+ +1 )^2}{N}  \le  \frac{C (\cN_+ +1 )^3}{N^2}. 
\end{align*}
Finally, consider 
$$
X:= \Big( \frac{1}{2} \sum_{p\ne 0} \widehat w(p) a_p^* a_{-p}^* \Big(1 +  \frac{1-\cN_+}{N-1} \Big)  + h.c.\Big) + \Big(\frac{1}{\sqrt{N}} \sum_{\ell,p\neq 0, \ell+p\neq 0} \widehat{w}(\ell) a^*_{p+\ell} a^*_{-\ell} a_{p} +\,\, \text{h.c.} \Big). 
$$
By the Cauchy-Schwarz inequality $\pm (Y^*Z+ Z^* Y)  \le Y^* Y + Z^* Z$ we can bound
\begin{align*}
\pm X &\le   \sum_{p\ne 0} a_p^* a_{-p}^* (\cN_+ +1)^{-1} a_{-p} a_p  +\sum_{p\ne 0} |\widehat w(p)|^2 (\cN_+ +1) \Big(1 +  \frac{1-\cN_+}{N-1} \Big)^2 \\
&\quad+ \frac{1}{N} \sum_{\ell,p\neq 0, \ell+p\neq 0} a^*_{p+\ell} a^*_{-\ell} a_{-\ell} a_{p+\ell}  +  \sum_{\ell,p\neq 0, \ell+p\neq 0} |\widehat{w}(\ell)|^2  a^*_{p} a_p \\
&\le C \Big[ (\cN_+ + 1) + \frac{(\cN_+ + 1)^2}{N} + \frac{(\cN_+ + 1)^3}{N^2} \Big].  
\end{align*}
Moreover, since $X$ changes the number of particles by at most 2, we have
\begin{align*}
X + \1^{>N} X \1^{>N} - \1^{\le N} X \1^{\le N} =   \1^{>N} X + X \1^{>N} = 1^{>N} X \1^{ > N-2}+ \1^{ > N-2} X \1^{>N}.
\end{align*}
Hence, combining with the above bound on $\pm X$ we find that 
\begin{align*}
\pm (X -   \1^{\le N} X \1^{\le N} ) &= \pm \Big( 1^{>N} X \1^{ > N-2}+ \1^{ > N-2} X \1^{>N}  -  \1^{>N} X \1^{>N}   \Big)\\
&\le C \Big[ (\cN_+ + 1) + \frac{(\cN_+ + 1)^2}{N} + \frac{(\cN_+ + 1)^3}{N^2} \Big] \1^{ > N-2} \\
& \le \frac{C (\cN_+ +1)^3}{N^2}. 
\end{align*}
 This completes the proof of the operator inequality on $\cF_+$ in Lemma \ref{lem:Bog-app}. 
\end{proof}

\section{Quadratic transformation} \label{sec:UB}

Recall that the Bogoliubov transformation $U_B$ in \eqref{eq:UB-intro} diagonalizes $\bH_{\rm Bog}$ as in \eqref{eq:Bog-UBHUB}. In this section, we will study the action of $U_N$ on the operator $\cG_N$. We have 

\begin{lemma}\label{lem:quadratic-transformation} Let $\cG_N$ be given in Lemma \ref{lem:Bog-app}. Then we have the operator identity on $\cF_{+}$ 
$$
U_B \cG_N U_B^* = \langle0| U_B \cG_N U_B^* |0\rangle + \sum_{p\ne 0} e(p) a_p^* a_p +  \cC_N  + R_2
$$
where
$$
\cC_N= \frac{1}{\sqrt{N}}\sum_{\substack{p,q\neq 0\\ p+q\neq 0}}\widehat{w}(p) \Big[ (\sigma_{p+q}\sigma_{-p}\gamma_q +\gamma_{p+q}\gamma_{p}\sigma_q ) a^*_{p+q} a^*_{-p} a^*_{-q} +\rm{h.c.} \Big]
$$
and the error term $R_2$ satisfies  
\begin{align*}
&\pm R_2 \le \frac{C}{\sqrt{N}} \cN_+^2 + \frac{C (\cN_+ +1)^3}{N^{3/2}}.   
\end{align*}
\end{lemma}

\begin{proof} Let us decompose
$$\cG_N- \bH_{\rm Bog}=\widetilde \cD_N + \widetilde \cC_N $$ 
where 
$$\widetilde \cC_N = \frac{1}{\sqrt{N}} \sum_{\ell,p\neq 0, \ell+p\neq 0} \widehat{w}(\ell) a^*_{p+\ell} a^*_{-\ell} a_{p} +\,\, \text{h.c.}
$$
{\bf Non-cubic terms.} Let us prove that $U_B \widetilde \cD_N U_B^*-  \langle0| U_B \widetilde \cD_N U_B^* |0\rangle$ contains only the terms of the form 
\begin{align}
\sum_{m_1,...,m_s,n_1,...,n_t \ne 0 }  A_{m_1,...,m_s, n_1,...,n_t} a_{m_1} ^* ... a_{m_s}^*  a_{n_1}... a_{n_t}   \label{eq:Amn}
\end{align}
with $1\le s+t \le 4$ and the coefficients $A_{m_1,...,m_s, n_1,...,n_t}$ satisfy 
\begin{align} \sup_{m_1,...,m_s\ne 0} \sum_{n_1,...,n_t \ne 0} |A_{m_1,...,m_s, n_1,...,n_t}|  \le \frac{C}{N}, \quad  \sup_{n_1,...,n_t \ne 0} \sum_{m_1,...m_s\ne 0} |A_{m_1,...,m_s, n_1,...,n_t}| \le \frac{C}{N}. \label{eq:Amn-1}
\end{align}

Let us start with the quadratic terms involving $a_p^* a_{-p}^*$. Using \eqref{eq:Bog-trans} and the CCR \eqref{eq:CCR} we have
\begin{align}
 & U_{\rm B} \Big( \frac{1}{4N}\sum_{p\ne 0} \widehat{w}(p) a_p^* a^*_{-p} \Big) U_{\rm B}^* =  \frac{1}{4N} \sum_{p\ne 0}  \widehat{w}(p) (\sigma_p a_p^* + \gamma_p a_{-p}) (\sigma_p a^*_{-p} + \gamma_{p} a_{p})  \nn\\
&= \frac{1}{4N} \sum_{p\ne 0}  \widehat{w}(p)  \Big[ \sigma_p^2 a_p^* a^*_{-p} +  2\sigma_p \gamma_{p} a_p^* a_{p}  + \gamma_p^2 a_{-p} a_{p}+ \sigma_p \gamma_{p} \Big].\label{eq:UBaaUB}
\end{align}
Obviously the constant in \eqref{eq:UBaaUB} satisfies 
$$
\frac{1}{4N} \sum_{p\ne 0}  \widehat{w}(p) \sigma_p \gamma_{p} = \Big\langle  0 \Big| U_{\rm B} \Big( \frac{1}{4N}\sum_{p\ne 0} \widehat{w}(p) a_p^* a^*_{-p} \Big) U_{\rm B}^* \Big| 0 \Big\rangle.
$$
Moreover, the other terms in \eqref{eq:UBaaUB} can be rewritten as
\begin{align}
 \frac{1}{4N} \sum_{p\ne 0}  \widehat{w}(p)  \sigma_p^2 a_p^* a^*_{-p} &= \frac{1}{4N} \sum_{p,q}  \widehat{w}(p)  \sigma_p^2 \delta_{p=-q}a_p^* a^*_{q},  \label{eq:UBaaUB-2a}\\
\frac{1}{2N} \sum_{p\ne 0}  \widehat{w}(p)  \sigma_p \gamma_{p} a_p^* a_{p}  &= \frac{1}{2N}  \sum_{p,q}  \widehat{w}(p)  \sigma_p \gamma_{p} \delta_{q=p} a_p^* a_{q} , \label{eq:UBaaUB-2b} \\
 \frac{1}{4N} \sum_{p\ne 0}  \widehat{w}(p)   \gamma_p^2 a_{-p} a_{p} &=  \frac{1}{4N} \sum_{p,q}  \widehat{w}(p)   \gamma_p^2 \delta_{p=-q} a_{p} a_{q} .\label{eq:UBaaUB-2c}
\end{align}
All of the sums in \eqref{eq:UBaaUB-2a}, \eqref{eq:UBaaUB-2b}, \eqref{eq:UBaaUB-2c} are of the general form \eqref{eq:Amn}-\eqref{eq:Amn-1}, thanks to the uniform bounds \eqref{eq:sigma-gamma-w}. The quadratic terms involving $a_p^* a_{p}$ can be treated similarly.

Next, consider 
\begin{align}
 & U_{\rm B} \Big( \frac{1}{2N}\sum_{p \ne 0 } \widehat{w}(p) a_p^* a^*_{-p} \cN_+ \Big) U_{\rm B}^* = U_{\rm B} \Big( \frac{1}{2N}\sum_{p,q\ne 0} \widehat{w}(p) a_p^* a^*_{-p} a_q^* a_q  \Big) U_{\rm B}^*\nn\\
 &= \frac{1}{2N} \sum_{p,q\ne 0}  \widehat{w}(p) (\sigma_p a_p^* + \gamma_p a_{-p}) (\sigma_p a^*_{-p} + \gamma_{p} a_{p}) (\sigma_q a^*_q + \gamma_q a_{-q}) (\sigma_q a_q + \gamma_q a^*_{-q})  \nn\\
 &= \frac{1}{2N} \sum_{p,q\ne 0}  \widehat{w}(p) \Big[ \sigma_p^2 a_p^* a^*_{-p} +  2\sigma_p \gamma_{p} a_p^* a_{p} + \gamma_p^2 a_{-p} a_{p}+ \sigma_p \gamma_{p} \Big]\times\nn\\
 &\qquad \qquad\qquad \qquad \times \Big[ (\sigma_q^2+\gamma_q^2) a_q^* a_q + \sigma_q \gamma_q (a_q^* a^*_{-q} + a_{-q} a_{q}) + \gamma_q^2 \Big]  \nn\\
 &=\frac{1}{2N} \sum_{p,q\ne 0}  \widehat{w}(p) \Big[ \sigma_p^2 a_p^* a^*_{-p} + \sigma_p \gamma_{p} \Big] \Big[ (\sigma_q^2 +\gamma_q^2) a_q^* a_q + \sigma_q \gamma_q (a_q^* a^*_{-q} + a_{-q} a_{q}) + \gamma_q^2 \Big] \nn \\
  &\quad +  \frac{1}{N} \sum_{p,q}  \widehat{w}(p) \sigma_p \gamma_{p} a_p^* \Big[ (\sigma_q^2 +\gamma_q^2) a_q^* a_q  + \sigma_q \gamma_q (a_q^* a^*_{-q} + a_{-q} a_{q}) +  \gamma_q^2 \Big] a_p \nn\\
   &\quad +  \frac{1}{N} \sum_{p,q\ne 0}  \widehat{w}(p) \sigma_p \gamma_{p}  \Big[ (\sigma_q^2 +\gamma_q^2)    a_p^* a_p \delta_{p,q} +   \sigma_q \gamma_q a_p^* a_{-p}^* (\delta_{p,q} +\delta_{p,-q})  \Big]  \nn\\
 &\quad +\frac{1}{2N} \sum_{p,q}  \widehat{w}(p)  \gamma_{p}^2 \Big[ (\sigma_q^2 +\gamma_q^2) a_q^* a_q  + \sigma_q \gamma_q (a_q^* a^*_{-q} + a_{-q} a_{q}) + \gamma_q^2 \Big]  a_p a_{-p} \nn\\
 &\quad + \frac{1}{2N} \sum_{p,q\ne 0}  \widehat{w}(p)  \gamma_{p}^2 \Big[ (\sigma_q^2 +\gamma_q^2) a_p a_{-p} (\delta_{p,q}+\delta_{p,-q})  \nn\\
 &\qquad \qquad \qquad \qquad + \sigma_q \gamma_q (a_p^* a_p + a_{-p}^* a_{-p} + 1) (\delta_{p,q}+ \delta_{p,-q}) \Big].  
   \label{eq:UBaaNUB}
\end{align}
It is straightforward to see that, except the constant
\begin{align*}
\frac{1}{2N} \sum_{p,q\ne 0}  \widehat{w}(p)  \sigma_p \gamma_{p} \gamma_q^2 + \frac{1}{2N} \sum_{p,q\ne 0}  \widehat{w}(p)  \gamma_{p}^2 \sigma_q \gamma_q  & (\delta_{p,q}+ \delta_{p,-q})  \\
&= \Big\langle 0 \Big| U_{\rm B} \Big( \frac{1}{2N}\sum_{p} \widehat{w}(p) a_p^* a^*_{-p} \cN_+ \Big) U_{\rm B}^* \Big| 0 \Big\rangle,
\end{align*}
all other terms in \eqref{eq:UBaaNUB} can be written as in  \eqref{eq:Amn}, with the corresponding bound \eqref{eq:Amn-1} following from \eqref{eq:sigma-gamma-w}. By the same argument, we can show that the terms involving $a_p^* a_{p}\cN_+$, $a^*_{p+\ell} a^*_{q-\ell} a_p a_q$ and $\cN_+ (\cN_+-1)$ are of the general form \eqref{eq:Amn}-\eqref{eq:Amn-1}.

Next, let us bound the terms of the general form \eqref{eq:Amn}-\eqref{eq:Amn-1}. We consider the case $s\ge t$ (the other case is treated similarly). By the Cauchy-Schwarz inequality 
$$Y^*Z+Z^*Y\le Y^* Y + Z^*Z$$ we have
\begin{align}\label{eq:generalN_+bound}
&\pm \Big( \sum_{m_1,...,m_{s}, n_1,...,n_{t}} A_{m_1,...,m_{s}, n_1,...,n_{t}} a_{m_1}^*\ldots a_{m_{s}}^* a_{n_1}\ldots a_{n_{t}} + h.c.  \Big) \nn \\
&\le \eps^{-1}  \sum_{m_1,...,m_{s}, n_1,...,n_{t}} |A_{m_1,...,m_{s}, n_1,...,n_{t}}| a_{m_1}^*\ldots a_{m_{s}}^* (\cN_+ +5)^{1-s} a_{m_{s}}\ldots a_{m_1} \nn \\
&\quad + \eps \sum_{m_1,...,m_{s}, n_1,...,n_{t}} |A_{m_1,...,m_{s}, n_1,...,n_{t}}| a_{n_{t}}^* \ldots a_{n_1}^* (\cN_+ +5)^{s-1} a_{n_1}\ldots a_{n_{t}}\nn\\
&\le \eps^{-1} \Big( \sup_{m_1,...,m_{s}} \sum_{n_1,...,n_{t}} |A_{m_1,...,m_{s}, n_1,...,n_{t}}| \Big) \sum_{m_1,...,m_{s}}  a_{m_1}^*\ldots a_{m_{s}}^* (\cN_+ +5)^{1-s} a_{m_{s}}\ldots a_{m_1} \nn\\
&\quad + \eps \Big( \sup_{n_1,...,n_{t'}} \sum_{m_1,...,m_{s}} |A_{m_1,...,m_{s}, n_1,...,n_{t}}| \Big) \sum_{n_1,...,n_{t'}}  a_{n_{t}}^* \ldots a_{n_1}^* (\cN_+ +5)^{s-1} a_{n_1}\ldots a_{n_{t}} \nn \\
&\le \eps^{-1}\frac{C}{N} \cN_+  + \eps \frac{C}{N} (\cN_+ + 1) ^{t+s-1}
\end{align}
for all $\eps>0$.  Note that if $\min(t,s) \ge 1$, then on the right side of \eqref{eq:generalN_+bound} we can replace $(\cN_++1)^{t+s-1}$ by $\cN_+^{t+s-1}$. 

In particular, for the non-cubic term $\widetilde \cD_N$, using \eqref{eq:generalN_+bound} with $\eps=N^{-1/2}$ and $t+s\le 4$ we get 
\begin{align}\label{eq:DN-last}
\pm \Big( U_B \widetilde \cD_N U_B^*- \langle0| U_B \widetilde \cD_N U_B^* |0\rangle \Big) \le \frac{C}{\sqrt{N}} \cN_+  + \frac{C(\cN_+ +1)^3}{N^{3/2}}.
\end{align}

\noindent
{\bf Cubic terms.} By using \eqref{eq:Bog-trans}  we have 
\begin{align}
&U_B \left(\frac{1}{\sqrt{N}} \sum_{\ell,p\neq 0, \ell+p\neq 0} \widehat{w}(\ell) a^*_{p+\ell} a^*_{-\ell} a_{p}  \right)  U^*_B \nn\\
&= \frac{1}{\sqrt{N}} \sum_{\ell,p\neq 0, \ell+p\neq 0} \widehat{w}(\ell) (\sigma_{p+\ell}a^*_{p+\ell} + \gamma_{p+\ell} a_{-p-\ell}) (\sigma_{\ell} a^*_{-\ell} + \gamma_{\ell} a_{\ell}) (\sigma_p a_{p} + \gamma_{p} a_{-p}^*)\nn\\
&=  \frac{1}{\sqrt{N}}\sum_{\substack{\ell,p\neq 0\\ \ell+p\neq 0}}\widehat{w}(\ell)  \Big( \sigma_{p+\ell}\sigma_{\ell}\gamma_p a^*_{p+\ell} a^*_{-\ell} a^*_{-p} + \gamma_{p+\ell}\gamma_{\ell}\sigma_p a_{-p-\ell} a_{\ell} a_p\Big) \nn\\
&\quad +  \frac{1}{\sqrt{N}}\sum_{\substack{\ell,p\neq 0\\ \ell+p\neq 0}}  \Big( \hat{w}(\ell)\sigma_{p+\ell}\sigma_{\ell}\sigma_p+\hat{w}(p+\ell)\sigma_{p+\ell}\gamma_{\ell}\gamma_{-p}+\hat{w}(p)\sigma_{p+\ell}\gamma_p \gamma_\ell \Big) a^*_{p+\ell} a^*_{-\ell} a_p \nn\\
&\quad +  \frac{1}{\sqrt{N}}\sum_{\substack{\ell,p\neq 0\\ \ell+p\neq 0}} \Big( \hat{w}(\ell)\sigma_{p+\ell}\gamma_{\ell}\sigma_p+\hat{w}(p+\ell)\sigma_{p+\ell}\sigma_{\ell}\gamma_{p}+\hat{w}(p)\gamma_{-p-\ell} \gamma_p \gamma_\ell  \Big) a_{p+\ell}^* a_p a_\ell. \label{UBCNUB}
\end{align}
By using \eqref{eq:sigma-gamma-w}, we can write the last sum of \eqref{UBCNUB}  as
$$
\sum_{p,q,r} \wt{A}_{p,q,r} a_p^* a_q a_r 
$$
with 
$$ \sup_{p} \sum_{q,r} |\wt{A}_{p,q,r}|  \le \frac{C}{\sqrt{N}}, \quad  \sup_{q,r} \sum_{p} |\wt{A}_{p,q,r}| \le \frac{C}{\sqrt{N}}.$$
Using \eqref{eq:generalN_+bound} with $\eps=1$, $t=1,s=2$, we get 
\begin{align}\label{eq:generalN_+bound-cubic}
&\pm \Big( \sum_{p,q,r} \wt{A}_{p,q,r} a_p^* a_q a_r  + h.c. \Big) \le  \frac{C}{N} (\cN_+ + \cN_+^2)\le   \frac{C \cN_+^2}{N}.  
\end{align}
Here $\cN_+\le \cN_+^2$ since the spectrum of $\cN_+$ is $\{0,1,2,...\}$. The second sum on the right side of \eqref{UBCNUB} can be treated by the same way. Thus from \eqref{UBCNUB} and its adjoint, we have
\begin{align*} 
&\pm \Big( U_B \widetilde \cC_N U_B^* -  \frac{1}{\sqrt{N}}\sum_{\substack{\ell,p\neq 0\\ \ell+p\neq 0}}\widehat{w}(\ell)  \Big[ \Big( \sigma_{p+\ell}\sigma_{\ell}\gamma_p a^*_{p+\ell} a^*_{-\ell} a^*_{-p} + \gamma_{p+\ell}\gamma_{\ell}\sigma_p a_{-p-\ell} a_{\ell} a_p\Big) + h.c.\Big] \Big)  \\
&\le \frac{C\cN_+^2}{N} 
\end{align*}
which is equivalent to
\begin{align} \label{eq:CN-last}
& \pm \Big( U_B \widetilde \cC_N U_B^* -  \cC_N\Big)  \le \frac{C\cN_+^2}{N} . 
\end{align}
In particular, \eqref{eq:CN-last} implies that 
$$\langle 0| U_B \widetilde \cC_N U_B^* |0\rangle =0.$$
Therefore, from \eqref{eq:Bog-UBHUB},  \eqref{eq:DN-last} and \eqref{eq:CN-last} we obtain the desired conclusion of  Lemma \ref{lem:quadratic-transformation}. 
\end{proof}

\section{Cubic transformation} \label{sec:US} 

To factor out the cubic term $\cC_N$ in Lemma \ref{lem:quadratic-transformation}, we will use a cubic renormalization. We will prove 
\begin{lemma}\label{lem:cubic-transformation} Let $\cC_N$ be the cubic term in Lemma \ref{lem:quadratic-transformation} and let $ U_S$ be given in \eqref{eq:US-intro}. Then we have the operator identity on Fock spacee $\cF_+$
$$
U_S U_B \cG_N U_B^* U_S^*= \Big\langle 0 \Big| U_S U_B \cG_N U_B^*  U_S^* \Big| 0 \Big\rangle + \sum_{p\ne 0} e(p) a_p^* a_p +  R_3
$$
where
$$
\pm R_3 \leq C\frac{\cN_+^2}{\sqrt{N}}  +\frac{C(\cN_++1)^4}{N^{3/2}} .  
$$
\end{lemma}
\begin{proof} Recall that from Lemma \ref{lem:quadratic-transformation} we have 
\bq \label{eq:USUB-00}
U_S U_B \cG_N U_B^* U_S^*= \Big\langle 0 \Big|   U_B \cG_N U_B^*   \Big| 0 \Big\rangle + U_S \Big( \dGamma(\xi) +  \cC_N \Big) U_S^* + U_S R_2 U_S^*
\eq
with 
$$
\dGamma(\xi) =\sum_{p\ne 0} e(p) a_p^* a_p,\quad \pm R_2 \le \frac{C}{\sqrt{N}} \cN_+^2 + C\frac{(\cN_+ +1)^3}{N^{3/2}}. 
$$
Thanks to  Lemma \ref{lem:moment-US} and Lemma \ref{lem:moment-US-2}, we find that
$$
\pm U_S R_2 U_S^* \le \frac{C}{\sqrt{N}} \cN_+^2 + C\frac{(\cN_+ +1)^3}{N^{3/2}}. 
$$
Thus this error term is part of $R_3$. 

For the main term, we use $U_S=e^{S}$ and the Duhamel formula 
\bq
e^{X}Y e^{-X}=Y+\int_0^1 e^{tX} [X,Y] e^{-tX}dt \label{eq:Duhamel}
\eq
we can write 
\begin{align}  \label{eq:USUB-01}
&e^{S} \Big( \dGamma(\xi)+ \cC_N\Big)   e^{-S}=\dGamma(\xi) + \cC_N + \int_0^1 e^{tS} \Big( [S,\dGamma(\xi)] + [S,\cC_N] \Big) e^{-tS} dt  \\
&= \dGamma(\xi) +  \int_0^1 e^{tS} \Big( \cC_N + [S,\dGamma(\xi)] + [S,\cC_N] \Big) e^{-tS} dt  - \int_0^1 \int_0^t e^{sS}  [S,\cC_N]  e^{-sS} ds dt \nn .
\end{align}
\noindent
{\bf Controlling $\cC_N + [S,\dGamma(\xi)]$.} Since $\dGamma(\xi)$ commutes with $\cN_+$ and
$$
[a_k^* a_k, a_{p+q}^* a^*_{-p} a^*_{-q}] = (\delta_{k,p+q} + \delta_{k,-p} + \delta_{k,-q}) a_{p+q}^* a^*_{-p} a^*_{-q}
$$
we find that 
\begin{align*} 
[\dGamma(\xi), S] &=    \frac{1}{\sqrt{N}}  \sum_{\substack{p,q\neq 0\\ p+q\neq 0}}  \quad \sum_{k\ne 0}  e(k)  \eta_{p,q}  [a_k^* a_k, a_{p+q}^* a^*_{-p} a^*_{-q}]  \1^{\le N}+ h.c.    \nn\\
& =    \frac{1}{\sqrt{N}} \sum_{\substack{p,q\neq 0\\ p+q\neq 0}}   (e(p+q)+e(p) + e(q) )  \eta_{p,q} a_{p+q}^* a^*_{-p} a^*_{-q} \1^{\le N}+ h.c.  \nn\\
&=  \frac{1}{\sqrt{N}} \sum_{\substack{p,q\neq 0\\ p+q\neq 0}}   \widehat{w}(p)\big(\sigma_{p+q}\sigma_p \gamma_q +\gamma_{p+q}\gamma_p \sigma_q\big) a_{p+q}^* a^*_{-p} a^*_{-q} \1^{\le N} + h.c. 
\end{align*}
which is equivalent to 
$$
\cC_N + [S,\dGamma(\xi)]  =  \frac{1}{\sqrt{N}} \sum_{\substack{p,q\neq 0\\ p+q\neq 0}}   \widehat{w}(p)\big(\sigma_{p+q}\sigma_p \gamma_q +\gamma_{p+q}\gamma_p \sigma_q\big) a_{p+q}^* a^*_{-p} a^*_{-q} \1^{>N} + h.c. 
$$
where $\1^{>N}=\1-\1^{\le N}=\1(\cN_+ >N)$. Thanks to the summability \eqref{eq:sigma-gamma-w}, we can use the Cauchy-Schwarz inequality similarly to  \eqref{eq:generalN_+bound} (with $\eps=1$) to get 
$$
\pm \Big( \cC_N + [S,\dGamma(\xi)] \Big) \le \frac{C}{\sqrt{N}}\cN_+ + \frac{C(\cN_+ +1)^2}{\sqrt{N}} \1^{>N} \le \frac{C}{\sqrt{N}}\cN_+ + \frac{C(\cN_+ +1)^3}{N^{3/2} }
$$
Combining with Lemma \ref{lem:moment-US} and Lemma \ref{lem:moment-US-2} we obtain 
\begin{align} \label{eq:USUB-02}
e^{tS} \Big( \cC_N + [S,\dGamma(\xi)] \Big)   e^{-tS} \le \frac{C}{\sqrt{N}}\cN_+ + \frac{C(\cN_+ +1)^3}{N^{3/2} }, \quad \forall t\in [-1,1]. 
\end{align}

\medskip

\noindent
{\bf Controlling $[S,\cC_N]$.} Let us decompose $S= \widetilde S - S^{>}$ where 
\begin{align*} 
\widetilde S &=\frac{1}{\sqrt{N}}\sum_{\substack{p,q\neq 0\\ p+q\neq 0}} \eta_{p,q}a^*_{p+q}a^*_{-p}a^*_{-q} - h.c.,\nn\\
S^{>}  &= \frac{1}{\sqrt{N}}\sum_{\substack{p,q\neq 0\\ p+q\neq 0}} \eta_{p,q}a^*_{p+q}a^*_{-p}a^*_{-q} \1^{>N}- h.c. 
\end{align*}
The main contribution comes from 
\begin{align*}
[\cC_N,\widetilde{S}]&= \frac{1}{N}\sum_{p,q,r\ne 0}   \sum_{p',q',r'\ne 0}  \delta_{p+q+r=0} \delta_{p'+q'+r'=0} \eta_{p',q'}   \widehat{w}(p) \big(\sigma_{r}\sigma_{p} \gamma_{q} +\gamma_{r}\gamma_{p} \sigma_{q}\big)  \times\\
&\qquad \times   \Big[ a_{r} a_{p} a_{q} , a_{r'}^* a_{p'} ^*a_{q'}^* \Big] + h.c. \\
&=  \frac{1}{N}\sum_{p,q,r\ne 0}   \sum_{p',q',r'\ne 0}  \delta_{p+q+r=0} \delta_{p'+q'+r'=0} \eta_{p',q'}   \widehat{w}(p) \big(\sigma_{r}\sigma_{p} \gamma_{q} +\gamma_{r}\gamma_{p} \sigma_{q}\big)   \times\\
&\qquad \times \Big(   \delta_{r=r'}  a_{p}a_q a_{p'} ^* a_{q'}^*  + \delta_{r=p'}  a_{p}a_q a_{r'} ^* a_{q'}^*  + \delta_{r=q'}  a_{p}a_q a_{r'} ^* a_{p'}^* \\
& \qquad  \qquad + \delta_{p=r'} a_q a_{p'}^* a_{q'}^* a_r  + \delta_{p=p'} a_{q} a_{r'}^* a_{q'}^* a_r   + \delta_{p=q'} a_{q} a_{r'}^* a_{p'}^* a_r \\ 
& \qquad  \qquad + \delta_{q=r'} a_{p'}^* a_{q'}^* a_r a_p   + \delta_{q=p'} a_{r'}^* a_{q'}^* a_r a_p  + \delta_{q=q'} a_{r'}^* a_{p'}^* a_r a_p \Big) + h.c. \\
\end{align*}
By using the CCR \eqref{eq:CCR} as in \eqref{eq:UBaaNUB}, we can write
$$
[\widetilde{S},\cC_N] = \langle 0 | [\widetilde{S},\cC_N] | 0\rangle +  \sum_{p,q\neq 0}A_{pq} a^*_p a_q+  \sum_{p,q,r,k\neq 0}B_{pqrk}a^*_p a^*_q a_r a_k 
$$
where 
\begin{align*}
\sup_{q\ne 0} \sum_{p\ne 0} |A_{pq} | &\le \frac{C}{N}, \quad \sup_{p\ne 0} \sum_{q\ne 0} |A_{pq} | \le \frac{C}{N}, \\
\sup_{p,q \ne 0} \sum_{r,k \ne 0} |B_{pqrk} | &\le \frac{C}{N}, \quad \sup_{r,k \ne 0} \sum_{p,q \ne 0} |B_{pqrk} | \le \frac{C}{N}. 
\end{align*}
By the Cauchy-Schwarz inequality as in \eqref{eq:generalN_+bound}, we get
\bq \label{eq:tilS-C}
\pm \Big( [\widetilde{S},\cC_N] - \langle 0 | [\widetilde{S},\cC_N] |0\rangle\Big) \le C\frac{\cN_+}{\sqrt{N}} + \frac{C(\cN_+ +1)^3}{N^{3/2}}. 
\eq
It remains to bound $[S^>, \cC_N]$. From the explicit form of $S^{>}$ and $\cC_N$, it is straightforward to check that 
$$
\pm [S^>, \cC_N]  = \pm \Big( (S^> ) \cC_N + \cC_N (S^>)^* \Big) \le (S^> )  (S^>)^* + \cC_N^2 \le \frac{C}{N} (\cN_+ +1)^3. 
$$
On the other hand, we observe that 
$$S^> = \1^{>N-4}  S^> \1^{>N-4} $$
and that $\cC_N$ does not change the number of particles more than 3. Therefore,  
\bq \label{eq:S>C}
\pm [S^>, \cC_N] = \pm \1^{> N-7} [S^>, \cC_N] \1^{> N-7} \le  \frac{C}{N} (\cN_+ +1)^3 \1^{> N-7} \le  \frac{C}{N^2} (\cN_+ +1)^4. 
\eq
Moreover, it is obvious that
$$
\langle 0|  [S^>,\cC_N] |0\rangle =0
$$
for $N\ge 10$. Thus from  \eqref{eq:tilS-C} and \eqref{eq:S>C} we obtain 
$$
\pm \Big( [S,\cC_N] - \langle 0 | [S,\cC_N] |0\rangle\Big) \le C\frac{\cN_+}{\sqrt{N}} + \frac{C(\cN_+ +1)^4}{N^{3/2}}.
$$
Combining with Lemma \ref{lem:moment-US} we conclude that
\begin{align}  \label{eq:USUB-03}
\pm e^{tS}\Big( [S,\cC_N] - \langle 0 | [S,\cC_N] |0\rangle\Big) e^{-tS}\le  C\frac{\cN_+}{\sqrt{N}} + \frac{C(\cN_+ +1)^4}{N^{3/2}}, \quad \forall t\in [-1,1].  
\end{align}

\noindent{\bf Conclusion.} Inserting \eqref{eq:USUB-02} and \eqref{eq:USUB-03} in \eqref{eq:USUB-01} we find that
$$
\pm \Big(e^{S} \Big( \dGamma(\xi)+ \cC_N\Big)   e^{-S} - \dGamma(\xi)  - \frac{1}{2} \langle 0 | [S,\cC_N] |0\rangle \Big) \le C\frac{\cN_+}{\sqrt{N}} + \frac{C(\cN_+ +1)^4}{N^{3/2}}.
$$
Combining with  \eqref{eq:USUB-00} we deduce that 
$$
\pm \Big( U_SU_B\cG_N U_B^* U_S^* - \dGamma(\xi) - \Big\langle 0 \Big|   U_B \cG_N U_B^*   \Big| 0 \Big\rangle - \frac{1}{2} \langle 0 | [S,\cC_N] |0\rangle  \Big) \le C\frac{\cN_+^2}{\sqrt{N}} + \frac{C(\cN_+ +1)^4}{N^{3/2}}.
$$
Taking the expectation of the latter bound again the vacuum, we find that
$$
\pm \Big( \langle 0 | U_SU_B\cG_N U_B^* U_S^* |0\rangle - \Big\langle 0 \Big|   U_B \cG_N U_B^*   \Big| 0 \Big\rangle - \frac{1}{2} \langle 0 | [S,\cC_N] |0\rangle  \Big) \le  \frac{C}{N^{3/2}}.
$$
Thus we obtain the desired conclusion
$$
\pm \Big( U_SU_B\cG_N U_B^* U_S^* - \dGamma(\xi) - \langle 0 | U_SU_B\cG_N U_B^* U_S^* |0\rangle \Big) \le C\frac{\cN_+^2}{\sqrt{N}} + \frac{C(\cN_+ +1)^4}{N^{3/2}}. 
$$
This completes the proof of Lemma \ref{lem:cubic-transformation}. 
\end{proof}

\section{Proof of Theorem \ref{thm:main-2}} \label{sec:Thm-2}

\begin{proof} We will prove the ground state energy estimate
$$
E_N=  \frac{N}{2}\widehat w(0)  + \Big\langle 0 \Big| U_S U_B \cG_N U_B^*  U_S^* \Big| 0 \Big\rangle + O(N^{-3/2}).
$$
\noindent
{\bf Upper bound.} We use the following $N$-body trial state
$$\widetilde \Psi_N= \frac{1}{\|U_N^* \1^{\le N} U_{\rm B}^* U_S^*|0\rangle\|} U_N^* \1^{\le N} U_{\rm B}^* U_S^*|0\rangle.$$
Then by the variational principle and the operator inequality on $\cF_+$ in Lemma \ref{lem:Bog-app} we have 
\begin{align*}
E_N &\le \langle \widetilde \Psi_N, H_N \widetilde \Psi_N\rangle = \frac{1}{\|U_N^* \1^{\le N} U_{\rm B}^* U_S^*|0\rangle\|^2}  \Big\langle 0\Big| U_S U_{\rm B} \1^{\le N}  U_N H_N U_N^*  \1^{\le N} U_{\rm B}^*U_S^*\Big|0\Big\rangle \\
&\le  \frac{1}{\|U_N^* \1^{\le N} U_{\rm B}^* U_S^*|0\rangle\|^2}   \Big\langle 0\Big|U_S U_{\rm B}  \Big ( \frac{N \widehat w(0)}{2}  + \cG_N +  \frac{C(\cN_+ +1)^3}{N^{3/2}} \Big) U_{\rm B}^*U_S^*\Big|0\Big\rangle . 
\end{align*}
By Lemma \ref{lem:moment-UB} and  Lemma \ref{lem:moment-US} we know that 
$$
\Big\langle 0\Big|U_S U_{\rm B}   (\cN_+ +1)^3 U_{\rm B}^*U_S^*\Big|0\Big\rangle \le C. 
$$
Consequently, 
\begin{align*} 
\|U_N^* \1^{\le N} U_{\rm B}^* U_S^*|0\rangle\|^2 &=1-\langle 0|U_S U_{\rm B}\1^{> N} U_{\rm B}^* U_S^*|0\rangle \\
&\ge 1 - \langle 0|U_S U_{\rm B} (\cN_+^3/N^3) U_{\rm B}^* U_S^*|0\rangle  \ge 1-CN^{-3}.
\end{align*}
Combining with Lemma \ref{lem:cubic-transformation} we find that
\begin{align*}
E_N &\le \frac{1}{\|U_N^* \1^{\le N} U_{\rm B}^* U_S^*|0\rangle\|^2}  \Big(  \frac{N \widehat w(0)}{2}  + \Big\langle 0\Big|U_S U_{\rm B}    \cG_N U_{\rm B}^*U_S^*\Big|0\Big\rangle + \frac{C}{N^{3/2}} 
\Big) \\
&\le \frac{N \widehat w(0)}{2}  + \Big\langle 0\Big|U_S U_{\rm B}    \cG_N U_{\rm B}^*U_S^*\Big|0\Big\rangle + \frac{C}{N^{3/2}} . 
\end{align*}
In the last estimate, we have also used the simple upper bound 
\bq \label{eq:upper-missing}
\Big\langle 0\Big|U_S U_{\rm B}    \cG_N U_{\rm B}^*U_S^*\Big|0\Big\rangle\le C
\eq 
which will be justified below. 

\bigskip

\noindent {\bf Lower bound.} Let $\Psi_N$ be the ground state of $H_{N}$ and denote $
\Phi := U_S U_B U_N \Psi_N \in \cF_+.$  By Lemmas \ref{lem:cN+},  \ref{eq:moment-UB} and \ref{eq:moment-US}, we have
$$
\langle \Phi, (\cN_+ +1)^4 \rangle \Phi\rangle \le C. 
$$
Then from  the operator identity on $\cF_+^{\le N}$ in Lemma \ref{lem:Bog-app} it follows that   
\begin{align*}
E_N &= \langle \Psi_N, H_N \Psi_N\rangle = \langle U_N \Psi_N, U_N  H_N  U_N^*  U_N\Psi_N\rangle \\
&\ge  \langle U_N \Psi_N, \Big( \frac{N \widehat w(0)}{2}  + \cG_N -  C \frac{(\cN_+ +1)^3 }{N^{3/2}}  \Big)  U_N\Psi_N\rangle \\
&\ge  \frac{N \widehat w(0)}{2}  + \langle \Phi, U_S U_B \cG_N U_B^* U_S^* \Phi\rangle  - CN^{-3/2}. 
\end{align*}
Next, using Lemma \ref{lem:cubic-transformation} together with two simple estimates: 
$$
\sum_{p\ne 0} e(p) a_p^* a_p  \ge \Big( \inf_{q\ne 0} e(q) \Big) \sum_{p\ne 0} a_p^* a_p \ge (2\pi)^2 \cN_+
$$
and 
$$
 \frac{\cN_+^2}{\sqrt{N}} \le \eps \cN_+ + \frac{\cN_+^2}{\sqrt{N}} \1(\cN_+ > \eps \sqrt{N}) \le \eps \cN_+  + C_\eps  \frac{\cN_+^4}{N^{3/2}}
$$
for $\eps>0$ small (but independent  of $N$), we obtain 
$$
U_S U_B \cG_N U_B^* U_S^* \ge  \langle 0| U_S U_B \cG_N U_B^* U_S^* |0\rangle + \cN_+ - C\frac{(\cN_++1)^4}{N^{3/2}}. 
$$
Therefore, 
\begin{align*}
\langle \Phi, U_S U_B \cG_N U_B^* U_S^* \Phi\rangle \ge \langle 0|  U_S U_B \cG_N U_B^* U_S^* |0 \rangle +  \langle \Phi, \cN_+ \Phi \rangle - CN^{-3/2}. 
\end{align*}
Thus 
\bq \label{eq:EN-lower-semi}
E_N\ge  \frac{N \widehat w(0)}{2} +  \langle 0|  U_S U_B \cG_N U_B^* U_S^* |0 \rangle + \Big \langle \Phi, \cN_+  \Phi \Big\rangle - CN^{-3/2}. 
\eq
From \eqref{eq:EN-lower-semi}, since $\Big \langle \Phi, \cN_+  \Phi \Big\rangle \ge 0$ we obtain the desired energy lower bound
$$
E_N \ge  \frac{N \widehat w(0)}{2} +  \langle 0|  U_S U_B \cG_N U_B^* U_S^* |0 \rangle + O(N^{-3/2}).
$$
This and the obvious upper bound $E_N \le  \widehat w(0) (N/2)$ imply the simple estimate \eqref{eq:upper-missing}. Thus the matching energy upper bound is valid, and hence we conclude that 
\bq \label{eq:GSE-thm}
E_N =  \frac{N \widehat w(0)}{2} +  \langle 0|  U_S U_B \cG_N U_B^* U_S^* |0 \rangle + O(N^{-3/2}). 
\eq
{\bf Ground state estimates.} By comparing the ground state energy expansion \eqref{eq:GSE-thm} with the lower bound \eqref{eq:EN-lower-semi} we deduce that 
\bq \label{eq:Phi-N32}
 \langle \Phi, \cN_+  \Phi \rangle \le CN^{-3/2}. 
\eq
Let us write $\Phi=(\Phi_j)_{j=0}^\infty$ with $\Phi_j\in \gH_+^j$. We can choose a phase factor for $\Psi_N$ such that $\Phi_0 \ge 0$. Then 
$$
\| \Phi - |0\rangle \|^2 = |\Phi_0 -1|^2 \le 1- |\Phi_0|^2 = \sum_{j\ge 1} |\Phi_j|^2 \le  \langle \Phi, \cN_+  \Phi \rangle \le CN^{-3/2}. 
$$
Putting back the definition $\Phi= U_S U_B U_N \Psi_N$ we obtain the norm approximation
$$
\| U_N \Psi_N - U_B^* U_S^* |0\rangle \| = \| \Phi - |0\rangle \|^2  \le CN^{-3/2}.
$$
This completes the proof of Theorem \ref{thm:main-2}. 
\end{proof}

\section{Proof of Theorem \ref{thm:main-1}} \label{sec:Thm-1} 

\begin{proof} Let $\Psi_N$ be the ground state for $H_N$. As explained in the introduction, we will decompose
$$
\gamma_{\Psi_N}^{(1)} = P \gamma_{\Psi_N}^{(1)} P + Q \gamma_{\Psi_N}^{(1)} Q + P \gamma_{\Psi_N}^{(1)} Q + Q \gamma_{\Psi_N}^{(1)} P.
$$
{\bf Diagonal terms.} For $Q \gamma_{\Psi_N}^{(1)} Q$, recall from \cite[Theorem 2.2 (iii)]{LewNamSerSol-15}  that
$$
 U_N \Psi_N \to U_B^* |0\rangle 
$$
strongly in the quadratic form of $\bH_{\rm Bog}$ on $\cF_+$. Moreover, it is easy to see that
$$
\bH_{\rm Bog} \ge \frac{1}{2} \sum_{p\ne 0} |p|^2 a_p^* a_p - C \ge \cN_+ - C
$$
(see e.g. \cite[Proof of Theorem 1]{Nam-18}). Therefore, in the limit $N\to \infty$, 
\begin{align} \label{eq:QgQ-proof-a}
\Tr Q \gamma_{\Psi_N}^{(1)} Q &= \langle U_N \Psi_N, \cN_+ U_N \Psi_N\rangle \to  \langle 0| U_B \cN_+ U_B^* |0\rangle \nn \\
&=  \Big\langle 0 \Big| \sum_{p\ne 0} ( \sigma_p  a_p^* + \gamma_p a_{-p}) (\sigma_p a_p + \gamma_p a_{-p}^*)   \Big |0 \Big\rangle =  \sum_{p\ne 0} \gamma_p^2.
\end{align}
Here we have used Bogoliubov's transformation \eqref{eq:Bog-trans}. Similarly, for any $p,q\ne 0$ we have
\begin{align}  \label{eq:QgQ-proof-b}
\langle u_p, Q\gamma_{\Psi_N}^{(1)} Q u_q\rangle &= \langle U_N \Psi_N, a_p^* a_q U_N \Psi_N \rangle \to \langle 0| U_B (a_p^*  a_q) U_B^* |0\rangle\nn \\
&=  \langle  0|  ( \sigma_p  a_p^* + \gamma_p a_{-p}) (\sigma_q a_q + \gamma_p a_{-q}^*) |0 \rangle = \gamma_p^2 \delta_{p,q}.  
\end{align}
From  \eqref{eq:QgQ-proof-a} and \eqref{eq:QgQ-proof-b}, we conclude that 
\bq \label{eq:QgQ-proof}
Q\gamma_{\Psi_N}^{(1)}  Q  \to \sum_{p\ne 0} \gamma_p^2 |u_p\rangle \langle u_p| 
\eq
strongly in trace class.  Consequently,
$$
\Tr (P \gamma_{\Psi_N}^{(1)}  P) = N - \Tr Q\gamma_{\Psi_N}^{(1)}  Q = N - \sum_{p\ne 0} \gamma_p^2
$$
and hence 
$$
\Tr \Big| P \gamma_{\Psi_N}^{(1)}  P + Q \gamma_{\Psi_N}^{(1)} Q -  \Big( N - \sum_{p\ne 0} \gamma_p^2 \Big)  |u_0\rangle\langle u_0|  - \sum_{p\ne 0} \gamma_p^2 |u_p\rangle \langle u_p|  \Big| \to 0. 
$$
{\bf Off-diagonal terms}.  Let us prove that
\bq \label{eq:PQ-off}
\Tr \Big| P  \gamma_{\Psi_N}^{(1)} Q + Q  \gamma_{\Psi_N}^{(1)} P \Big| \le CN^{-1/4}. 
\eq
By using $P=|u_0\rangle \langle u_0|$ and the Cauchy-Schwarz inequality, it suffices to show that
$$
\| Q  \gamma_{\Psi_N}^{(1)}u_0\|^2 \le CN^{-1/2}. 
$$
Since $\{u_p\}_{p\ne 0}$ is an  orthonormal basis for $\gH_+$, we have  
$$
\| Q  \gamma_{\Psi_N}^{(1)} u_0\|^2 =\sum_{p\ne 0} |\langle u_p,  \gamma_{\Psi_N}^{(1)} u_0\rangle|^2 = \sum_{p\ne 0} |\langle \Psi_N, a^*_0 a_p \Psi_N \rangle|^2.
$$
Using the excitation map $U_N$ and the relations \eqref{eq:UN}  we can decompose 
\begin{align*}
  & \langle \Psi_N, a^*_0 a_p \Psi_N \rangle=  \langle U_N \Psi_N,  \sqrt{N-\cN_+} a_p U_N \Psi_N \rangle \\
 &=  \sqrt{N} \langle U_N \Psi_N,   a_p U_N \Psi_N \rangle + \langle U_N \Psi_N,  (\sqrt{N-\cN_+} -\sqrt{N})  a_p U_N \Psi_N \rangle. 
\end{align*}
Therefore,  by the Cauchy-Schwarz inequality 
\begin{align} \label{eq:1-pdm-01}
\| Q  \gamma_{\Psi_N}^{(1)} u_0\|^2 &\le   2 N \sum_{p\ne 0}  |\langle U_N \Psi_N,   a_p U_N \Psi_N \rangle|^2 \nn \\
   &\qquad +2  \sum_{p\ne 0} | \langle U_N \Psi_N,  (\sqrt{N-\cN_+} -\sqrt{N})  a_p  U_N \Psi_N \rangle|^2.
\end{align}

For the second sum in \eqref{eq:1-pdm-01}, using the Cauchy-Schwarz inequality, the simple bound  
$$(\sqrt{N-\cN_+} -\sqrt{N})^2 = \Big( \frac{\cN_+}{\sqrt{N-\cN_+}+\sqrt{N}}\Big)^2  \le  \frac{\cN_+^2}{N}$$
and  Lemma \ref{lem:cN+}, we find that 
\begin{align}\label{eq:1-pdm-01-a}
& \sum_{p\ne 0} | \langle U_N \Psi_N,  (\sqrt{N-\cN_+} -\sqrt{N})  a_p  U_N \Psi_N \rangle|^2 \nn\\
& \le \sum_{p\ne 0} \| (\sqrt{N-\cN_+} -\sqrt{N})  U_N \Psi_N\|^2 \| a_p U_N \Psi_N \|^2\nn \\
&\le N^{-1} \langle U_N \Psi_N,  \cN_+^2 U_N \Psi_N\rangle  \langle U_N \Psi_N, \cN_+ U_N \Psi_N\rangle  \nn\\
&=  N^{-1} \langle  \Psi_N, \cN_+^2   \Psi_N\rangle  \langle  \Psi_N, \cN_+   \Psi_N\rangle \le CN^{-1}. 
\end{align}

To control the first sum in \eqref{eq:1-pdm-01}, we will use the bound  from Theorem \ref{thm:main-2}: 
\bq \label{eq:Phi-key-abc}
\langle \Phi, \cN_+ \Phi\rangle \le CN^{-3/2}, \quad \text{with}\quad \Phi= U_S U_B U_N \Psi_N.
\eq
Also, from Lemma \ref{lem:cN+}, Lemma \ref{lem:moment-UB} and Lemma \ref{lem:moment-US} it follows that
\bq \label{eq:Phi-key-abc-2}
\langle \Phi, (\cN_+ +1)^4 \Phi\rangle \le C.  
\eq
Using the action of Bogoliubov's transformation in \eqref{eq:Bog-trans} and the uniform bounds \eqref{eq:sigma-gamma-w}  we obtain 
\begin{align}\label{eq:1-pdm-01-b1}
& \sum_{p \ne 0}|\langle U_N \Psi_N,   a_p U_N \Psi_N \rangle|^2 =  \sum_{p \ne 0} |\langle  \Phi,   U_S U_B a_p  U_B^* U_S^* \Phi \rangle|^2 \nn\\
&= \sum_{p \ne 0} |\langle   \Phi,  U_S (\sigma_p a_p + \gamma_p a_{-p}^*) U_S^* \Phi \rangle|^2 \le  C  \sum_{p \ne 0} |\langle   \Phi,  U_S a_p U_S^* \Phi \rangle|^2.  
\end{align}
To estimate further the right side of \eqref{eq:1-pdm-01-b1}, we use the Duhamel formula
\begin{align*}
U_S a_p U_S^* = a_p + \int_0^1 e^{tS} [S,a_p] e^{-tS} \d t 
\end{align*}
and the Cauchy-Schwarz inequality to get 
\begin{align}\label{eq:1-pdm-01-b2}
\sum_{p \ne 0} |\langle   \Phi,  U_S a_p U_S^* \Phi \rangle|^2 \le 2 \sum_{p \ne 0} |\langle   \Phi,  a_p  \Phi \rangle|^2 + 2 \sum_{p \ne 0} \int_0^1  |\langle   \Phi,  e^{tS} [S,a_p] e^{-tS}   \Phi \rangle|^2   \d t . 
\end{align}
Thanks to  \eqref{eq:Phi-key-abc} we can bound 
\bq \label{eq:key-ab0}
\sum_{p \ne 0} |\langle   \Phi,  a_p  \Phi \rangle|^2 \le \sum_{p \ne 0} \|a_p  \Phi \|^2 = \langle \Phi, \cN_+ \Phi\rangle \le CN^{-3/2}. 
\eq
It remains to handle the term involving the commutator $[S,a_p]$ in \eqref{eq:1-pdm-01-b2}. Using the CCR \eqref{eq:CCR} and the identity $[a_p, \1^{\le N}] = -\1(\cN_+=N) a_p$ we can decompose
\begin{align*}
[a_p, S] &= \frac{1}{\sqrt{N}} \sum_{\substack{m,n\ne 0\\ m+n \ne 0}} \eta_{m,n} \Big[a_p,a_{m+n}^* a_{-m}^* a_{-n}^* \1^{\le N} - \1^{\le N} a_{m+n}  a_{-m}  a_{-n}  \Big]  \nn\\
&= \frac{1}{\sqrt{N}}\sum_{\substack{m,n\ne 0\\ m+n \ne 0}} \eta_{m,n} (\delta_{p,m+n} a_{-m}^* a_{-n}^* + \delta_{p,-m} a_{m+n}^* a_{-n}^* + \delta_{p,-n} a_{m+n}^* a_n^*) \1^{\le N} \nn\\
& \quad - \frac{1}{\sqrt{N}}\sum_{\substack{m,n\ne 0\\ m+n \ne 0}} \eta_{m,n} a_{m+n}^* a_{-m}^* a_{-n}^* \1(\cN_+=N) a_p\nn\\
&\quad +  \frac{1}{\sqrt{N}}\sum_{\substack{m,n\ne 0\\ m+n \ne 0}} \eta_{m,n} \1(\cN_+=N) a_p a_{m+n}  a_{-m}  a_{-n}  \nn \\
& =: I_1(p) + I_2(p) + I_3(p). 
\end{align*}
Hence, by the Cauchy--Schwarz inequality we have for all $t\in [0,1]$, 
\begin{align} \label{eq:Phi-tS-ab}
\sum_{p\ne 0} |\langle  \Phi, e^{tS} [a_p,S]  e^{-tS} \Phi\rangle|^2 \le 3\sum_{k=1}^3 \sum_{p\ne 0} |\langle  \Phi, e^{tS} I_k(p)  e^{-tS} \Phi\rangle|^2. 
\end{align}
The right side of \eqref{eq:Phi-tS-ab} can be bounded using the Cauchy-Schwarz inequality, the summability \eqref{eq:US-summability}, Lemma \ref{lem:moment-US}, Lemma \ref{lem:moment-US-2}, \eqref{eq:Phi-key-abc} and \eqref{eq:Phi-key-abc-2}. For the terms involving $I_1(p)$, we have  
\begin{align*}
&\sum_{p\ne 0} |\langle  \Phi, e^{tS} I_1(p)  e^{-tS} \Phi\rangle|^2 \\
&\le  \frac{C}{N}   \Big( \sum_{\substack{m,n,p\ne 0\\ m+n \ne 0}} |\eta_{m,n}|^2  \delta_{p,m+n} \Big)   \Big( \sum_{\substack{m,n\ne 0\\ m+n \ne 0}}  \|  (\cN_+ +3)^{-1/2} a_{-m} a_{-n}  e^{-tS} \Phi \|^2  \Big) \| (\cN_+ +3)^{1/2}  \1^{\le N} e^{ -tS} \Phi\|^2\\
& +  \frac{C}{N}   \Big( \sum_{\substack{m,n,p\ne 0\\ m+n \ne 0}} |\eta_{m,n}|^2  \delta_{p,-m} \Big)   \Big( \sum_{\substack{m,n\ne 0\\ m+n \ne 0}}  \|  (\cN_+ +3)^{-1/2} a_{m+n} a_{-n}  e^{-tS} \Phi \|^2  \Big) \| (\cN_+ +3)^{1/2}  \1^{\le N} e^{ -tS} \Phi\|^2 \\
& +  \frac{C}{N}   \Big( \sum_{\substack{m,n,p\ne 0\\ m+n \ne 0}} |\eta_{m,n}|^2  \delta_{p,-n} \Big)   \Big( \sum_{\substack{m,n\ne 0\\ m+n \ne 0}}  \|  (\cN_+ +3)^{-1/2} a_{m+n} a_{-m}  e^{-tS} \Phi \|^2  \Big) \| (\cN_+ +3)^{1/2}  \1^{\le N} e^{ -tS} \Phi\|^2\\
&\le \frac{C}{N} \Big\langle \Phi, e^{tS} \cN_+ e^{-tS} \Phi \Big\rangle \Big\langle \Phi, e^{tS} (\cN_+ +3) e^{ -tS} \Phi \Big\rangle\\
& \le \frac{C}{N} \left( \Big\langle \Phi, \cN_+ \Phi\Big\rangle + \frac{1}{N} \Big\langle \Phi, (\cN_+ +1)^2 \Phi\Big\rangle \right) \Big\langle \Phi, (\cN_+ +3) \Phi \Big\rangle \le \frac{C}{N^2}.
\end{align*}
Similarly,  the terms involving $I_2(p)$ are bounded by 
\begin{align*} 
&\sum_{p\ne 0} |\langle  \Phi, e^{tS} I_2(p)  e^{-tS} \Phi\rangle|^2 \nn \\
&\le  \frac{C}{N}   \Big( \sum_{\substack{m,n\ne 0\\ m+n \ne 0}} |\eta_{m,n}|^2  \Big)   \Big( \sum_{\substack{m,n\ne 0\\ m+n \ne 0}}  \| \1(\cN_+=N) a_{m+n} a_{-m} a_{-n}  e^{-tS} \Phi \|^2  \Big) \Big( \sum_{p\ne 0} \| a_p e^{ -tS} \Phi\|^2 \Big) \nn\\
&\le \frac{C}{N} \Big\langle \Phi, e^{tS} \cN_+^3 e^{-tS} \Phi \Big\rangle \Big\langle \Phi, e^{tS} \cN_+ e^{ -tS} \Phi \Big\rangle\nn\\
&\le  \frac{C}{N} \Big\langle \Phi, (\cN_+ +1)^3 \Phi \Big\rangle \left( \Big\langle \Phi,   \cN_+   \Phi \Big\rangle +\frac{1}{N} \Big\langle \Phi,  (\cN_++1)^2   \Phi \Big\rangle   \right)  \le \frac{C}{N^{2}}.  
\end{align*}
Finally for the terms involving $I_3(p)$, using 
$$
\1(\cN_+ = N) \le (\cN_+/N)^4 
$$
we have
\begin{align*} 
&\sum_{p\ne 0} |\langle  \Phi, e^{tS} I_3(p)  e^{-tS} \Phi\rangle|^2 \nn\\
&\le  \frac{C}{N}   \Big( \sum_{\substack{m,n\ne 0\\ m+n \ne 0}} |\eta_{m,n}|^2  \Big) \|    \1(\cN_+=N) e^{ -tS} \Phi\|^2  \Big( \sum_{\substack{m,n,p\ne 0\\ m+n \ne 0}}  \|  a_{m+n} a_{-m} a_{-n} a_p   e^{-tS} \Phi \|^2  \Big)  \nn\\
&\le \frac{C}{N} \Big\langle \Phi, e^{tS} (\cN_+/N)^4 e^{-tS} \Phi  \Big\rangle \Big\langle \Phi, e^{tS} \cN_+^4  e^{ -tS} \Phi \Big\rangle\nn\\
&\le  \frac{C}{N^5} \Big\langle \Phi, (\cN_+ +1)^4 \Phi \Big\rangle^2 \le \frac{C}{N^{5}}.  
\end{align*}
Thus  we conclude from \eqref{eq:Phi-tS-ab} that
\begin{align}\label{eq:Phi-tS-ab2}
\sum_{p\ne 0}  |\langle  \Phi, e^{tS} [S,a_p]  e^{-tS} \Phi\rangle|^2   \le \frac{C}{N^2}, \quad \forall t\in [0,1]. 
\end{align}
Consequently, 
\begin{align}\label{eq:Phi-tS-ab3}
\sum_{p\ne 0} \int_0^1  |\langle  \Phi, e^{tS} [S,a_p]  e^{-tS} \Phi\rangle|^2 \d t \le \frac{C}{N^2}.
\end{align}
Inserting \eqref{eq:Phi-tS-ab3} and \eqref{eq:key-ab0} in \eqref{eq:1-pdm-01-b2} and \eqref{eq:1-pdm-01-b1} we obtain 
$$
\sum_{p \ne 0}|\langle U_N \Psi_N,   a_p U_N \Psi_N \rangle|^2 \le C \sum_{p \ne 0} |\langle   \Phi,  U_S a_p U_S^* \Phi \rangle|^2 \le CN^{-3/2}. 
$$
Using the latter bound and \eqref{eq:1-pdm-01-a}, we deduce from \eqref{eq:1-pdm-01} that  
$$
\| Q  \gamma_{\Psi_N}^{(1)} u_0\|^2 \le CN^{-1/2}.
$$
This implies \eqref{eq:PQ-off} and completes the proof of Theorem \ref{thm:main-1}. 
\end{proof}

\end{document}